\documentclass[12pt,authoryear]
{article}
\usepackage[margin=2cm]{geometry}

\usepackage{authblk}
\usepackage{amsmath}
\usepackage{amssymb}
\usepackage{amsfonts}
\usepackage{cancel}
\usepackage{amsthm}
\usepackage{mathtools}
\usepackage{comment}
\allowdisplaybreaks

\usepackage[T1]{fontenc}
\usepackage[mathscr]{eucal}
\usepackage{dsfont}

\usepackage{fullpage}

\usepackage{setspace}

\usepackage{acronym}
\usepackage{enumitem}
\usepackage{latexsym}
\usepackage{longtable}
\usepackage{wasysym}
\usepackage{xspace}

\usepackage[dvipsnames,svgnames]{xcolor}
\colorlet{MyBlue}{DodgerBlue!60!Black}
\colorlet{MyGreen}{DarkGreen!85!Black}

\usepackage[font=small,labelfont=bf]{caption}
\usepackage{tikz}
\usetikzlibrary{positioning, arrows.meta}
\usetikzlibrary{calc,patterns}
\usepackage{graphicx}
\usepackage{subcaption}
\graphicspath{{journal/PlotSameMu/}}

\usepackage[authoryear,comma]{natbib}

\usepackage{hyperref}
\hypersetup{
colorlinks=true,
linktocpage=true,
pdfstartview=FitH,
breaklinks=true,
pdfpagemode=UseNone,
pageanchor=true,
pdfpagemode=UseOutlines,
plainpages=false,
bookmarksnumbered,
bookmarksopen=false,
bookmarksopenlevel=1,
hypertexnames=true,
pdfhighlight=/O,
urlcolor=MyBlue!60!black,
linkcolor=MyBlue!80!black,
citecolor=DarkGreen!90!black, 
pdftitle={},
pdfauthor={},
pdfsubject={},
pdfkeywords={},
pdfcreator={pdfLaTeX},
pdfproducer={LaTeX with hyperref}
}

\numberwithin{equation}{section}  
\usepackage[sort&compress,capitalize,nameinlink]{cleveref}
\crefname{app}{Appendix}{Appendices}

\crefrangeformat{equation}{\upshape(#3#1#4)\textendash(#5#2#6)}

\usepackage[textwidth=30mm]{todonotes}

\usepackage{soul}
\setstcolor{red}
\sethlcolor{SkyBlue}

\newcommand{\debug}[1]{{\color{purple}#1}}

\usepackage{verbatim}

\usepackage[linesnumbered,ruled,vlined]{algorithm2e}

\SetCommentSty{mycommfont} 
\SetAlgoSkip{bigskip} 

\theoremstyle{plain}
\newtheorem{theorem}{Theorem}
\newtheorem{corollary}[theorem]{Corollary}
\newtheorem{lemma}[theorem]{Lemma}
\newtheorem{proposition}[theorem]{Proposition}

\theoremstyle{definition}
\newtheorem{definition}[theorem]{Definition}

\theoremstyle{remark}
\newtheorem{remark}[theorem]{Remark}
\newtheorem{example}[theorem]{Example}

\numberwithin{theorem}{section}

\DeclarePairedDelimiter{\braces}{\{}{\}}
\DeclarePairedDelimiter{\bracks}{[}{]}
\DeclarePairedDelimiter{\parens}{(}{)}

\DeclarePairedDelimiter{\abs}{\lvert}{\rvert}
\DeclarePairedDelimiter{\infloor}{\lfloor}{\rfloor}

\DeclarePairedDelimiter{\ceil}{\lceil}{\rceil}
\DeclarePairedDelimiter{\floor}{\lfloor}{\rfloor}

\DeclarePairedDelimiterX{\braket}[2]{\langle}{\rangle}{#1,#2}

\DeclarePairedDelimiterX{\inner}[2]{\langle}{\rangle}{#1,#2}
\DeclarePairedDelimiterX{\setdef}[2]{\{}{\}}{#1:#2}

\DeclarePairedDelimiterXPP{\probof}[1]{\Prob}{(}{)}{}{%

#1}

\DeclarePairedDelimiterXPP{\exof}[1]{\Expect}{[}{]}{}{%

#1}

\newcommand{\naturals}{\mathbb{\debug N}}
\newcommand{\reals}{\mathbb{\debug R}}

\newcommand{\distr}{\debug F}

\DeclareMathOperator{\Expect}{\debug{\mathsf{E}}}
\DeclareMathOperator{\Prob}{\debug{\mathsf{P}}}

\DeclareMathOperator{\Exp}{\mathsf{\debug{Exp}}}

\DeclareMathOperator{\Poisson}{\mathsf{\debug{Poisson}}}

\newcommand{\Cost}{\debug C}
\newcommand{\Costtype}{\debug \Cost_{\type}}
\newcommand{\CostA}{\debug \Cost_{\typeA}}
\newcommand{\CostB}{\debug \Cost_{\typeB}}

\newcommand{\eq}[1]{#1^{\equ}}

\newcommand{\argdot}{\,\cdot\,}

\newcommand{\smalloh}{\debug o}

\newcommand{\ie}{i.e., }
\newcommand{\eg}{e.g., }
\newcommand{\iid}{i.i.d.\ }
\newcommand{\rhs}{r.h.s.\ }

\newcommand{\Acal}{\mathcal{\debug A}}
\newcommand{\AnB}{\Acal(\nBi)}
\newcommand{\approxerror}{\debug \delta}

\newcommand{\bound}{\debug M}

\newcommand{\const}{\debug c}
\newcommand{\constone}{\const_{1}}
\newcommand{\consttwo}{\const_{2}}
\newcommand{\cust}{\debug i}

\newcommand{\eqn}{\eq{\nAB}}
\newcommand{\eqthreshB}{\debug K}
\newcommand{\eqthreshBeps}{\eqthreshB_{\varepsilon}}

\newcommand{\error}{\mathcal{\debug E}}
\newcommand{\Etime}{\debug E} 
\newcommand{\Etimek}{\Etime^{\threshgen}} 
\newcommand{\EtimeNR}{\Etime^{\NR}}
\newcommand{\Eruin}{\mathcal{\debug E}}

\newcommand{\gena}{\debug a}
\newcommand{\genb}{\debug b}

\newcommand{\initnum}{\debug j}
\newcommand{\intercept}{\debug \alpha}

\newcommand{\MMone}[1][1]{M/M/$#1$\,}

\newcommand{\nA}{\nAB_{\typeA}} 
\newcommand{\nAi}{\nA^{(\cust)}}
\newcommand{\nB}{\nAB_{\typeB}} 
\newcommand{\nBi}{\nB^{(\cust)}}
\newcommand{\nAB}{\debug n}

\newcommand{\nAstar}{\debug{m}_{\typeA}}
\newcommand{\mAstar}{\debug{m}_{\typeA}^{*}}

\newcommand{\naor}{\debug N_{\typeA}} 

\newcommand{\nbarB}{\debug{\overline{N}_{\typeB}}}
\newcommand{\ncust}{\hat{\debug Z}} 

\newcommand{\ntype}{\nAB_{\type}}
\newcommand{\ntypei}{\ntype^{(\cust)}}

\newcommand{\presplay}{\debug S}
\newcommand{\presplayt}{\presplay_{\ttime}}
\newcommand{\Preward}{\debug P} 
\newcommand{\Prewardk}{\Preward^{\threshgen}}
\newcommand{\Pruin}{\mathcal{\debug P}} 

\newcommand{\rate}{\debug \lambda} 
\newcommand{\rateA}{\debug \rate_{\typeA}} 
\newcommand{\rateB}{\debug \rate_{\typeB}} 
\newcommand{\rateS}{\debug \mu} 
\newcommand{\rateSA}{\rateS_{\typeA}} 
\newcommand{\rateSB}{\rateS_{\typeB}} 

\newcommand{\ratetype}{\debug \rate_{\type}} 
\newcommand{\rateStype}{\rateS_{\type}} 

\newcommand{\reward}{\debug R} 
\newcommand{\rewardA}{\debug \reward_{\typeA}} 
\newcommand{\rewardB}{\debug \reward_{\typeB}} 
\newcommand{\rewardtype}{\reward_{\type}}

\newcommand{\run}{\debug N} 

\newcommand{\slope}{\debug \beta}

\newcommand{\state}{\debug \omega}

\newcommand{\strat}{\debug \sigma}
\newcommand{\stratk}{\strat^{\threshgen}}
\newcommand{\sumrho}{\debug s}
\newcommand{\sumrhok}{\sumrho_{\threshgen}}
\newcommand{\sums}{\debug d}
\newcommand{\sumsk}{\sums_{\threshgen}}

\newcommand{\threshgen}{\debug k}

\newcommand{\threshB}{\debug J}
\newcommand{\truncaff}{\debug G}

\newcommand{\ttime}{\debug t} 
\newcommand{\type}{\debug \theta} 
\newcommand{\typeA}{\debug A} 
\newcommand{\typeB}{\debug B} 

\newcommand{\utilizationA}{\debug \rho_{\typeA}} 
\newcommand{\utilizationtype}{\debug \rho_{\type}} 
\newcommand{\Utime}{\debug U} 

\newcommand{\valueB}{\debug V}
\newcommand{\valb}{\debug v}
\newcommand{\Veq}{ \debug V_{\typeB}} 

\newcommand{\bellman}{\debug U}

\newcommand{\Wtime}{\debug W}

\DeclareMathOperator{\equ}{\debug{eq}}
\DeclareMathOperator{\NR}{\debug{NR}}

\newcommand{\MS}[1]{\todo[color=blue!10!red!30,author=\textbf{Marco},inline]{\small #1\\}}

\newcommand{\acdef}[1]{\acfi{#1}\acused{#1}}

\newacro{FCFS}{first-come first-served}
\newacro{LCFS}{last-come fist-served}
\newacro{LCFSP}[LCFS-P]{last-come first-served with preemption}

\newacro{MPE}{Markov perfect equilibrium}
\newacroplural{MPE}[MPE]{Markov perfect equilibria}
\newacro{BRD}{best response dynamics}
\newacro{NE}{Nash equilibrium}
\newacroplural{NE}[NE]{Nash equilibria}
\newacro{PNE}{pure Nash equilibrium}
\newacroplural{PNE}[PNE]{pure Nash equilibria}
\newacro{MNE}{mixed Nash equilibrium}
\newacroplural{MNE}[MNE]{mixed Nash equilibria}
\newacro{PFNE}{prior-free Nash equilibrium}
\newacroplural{PFNE}[PFNE]{prior-free Nash equilibria}
\newacro{KKT}{Karush\textendash Kuhn\textendash Tucker}
\newacro{FIP}{finite improvement property}
\newacro{CLT}{central limit theorem}

\newacro{iid}[i.i.d.]{independent and identically distributed}

\newacro{ED}{emergency department}

\begin{document}

\title{Strategic Queues with Priority Classes}

\author[1]{Maurizio D'Andrea}
\affil[1]{Department of Analytics and Operations\\ National University of Singapore's Business School\\
\texttt{dmaurizio\_1@nus.edu.sg}}

\author[2]{Marco Scarsini}
\affil[2]{Department of Economics and Financial Markets, Luiss University\\
\texttt{marco.scarsini@luiss.it}
}

\maketitle

\begin{abstract}

We consider a strategic \MMone queueing model under a \acl{FCFS} regime, where customers are split into two classes and class $\typeA$ has priority over class $\typeB$. 
Customers can decide whether to join the queue or balk, and, in case they have joined the queue, whether and when to renege.
We model this situation as a game and study the equilibrium strategies for customers in the two classes.
We provide a recursive characterization of the symmetric
stationary equilibrium, allowing for different service rates across classes. When service rates are equal, we recover the result that the equilibrium depends only on the
total number of customers ahead. Finally, we establish
local stability under small perturbations of
service rates, show that linear approximations
may have arbitrarily large errors, and derive
explicit linear bounds on the equilibrium
threshold.



\end{abstract}

%
%
\section{Introduction}
\label{se:intro}

%
%

\subsection{The Problem}
\label{suse:problem}

Many real-life situations involve queues where customers make strategic decisions. 
For instance, customers choose which queue to join at the supermarket;
they decide when to go to a restaurant that does not take reservations;
they decide when to hang up after being put on hold when calling a call center;
they choose whether to pay for a  status that guarantees priority when boarding planes, etc. 
The analysis of strategic queueing models started with \citet{Nao:E1969}, who considered an \MMone queue under a \ac{FCFS} regime, where customers are homogeneous, receive a fixed reward for being served and pay a fixed cost for any unit of time spent in the system. 
Their decision is whether to join the queue or balk upon arrival.
He framed this problem as a game and showed that the unique equilibrium is given by a threshold strategy, \ie each customer joins the queue if and only if it is shorter than a certain threshold, which depends on the parameters of the model.
In this model, even if customers were allowed to renege, they would have no reason to do so. 

The homogeneity assumption is fundamental for Naor's result. 
In real life customers are often inhomogeneous and can differ in various aspects. 
For instance, they could be split into classes that are associated with different priorities, that is, high-priority customers are always ahead of low-priority customers in the queue.
Therefore, low-priority customers get served only when there are no high-priority customers in the system.
In this case, the equilibrium of the game is more complicated than in the homogeneous case.
In particular, whereas the high-priority customers are unaffected by the presence of low-priority customers and behave like in the Naor's model, low-priority customers use a threshold strategy that depends on the numbers of customers of each class in the system.

%
%

\subsection{Our Contribution}
\label{suse:contribution}

In this paper we consider an \MMone queueing model under a \ac{FCFS} regime, where customers belong to either class $\typeA$ or class $\typeB$.
$\typeA$-customers have priority over $\typeB$-customers.
Within each class customers are homogeneous. 
Customers of the two classes differ in their arrival rates, their service rates, their rewards, and their unit costs for being in the system. 
We characterize the equilibrium behavior of
both classes. A-customers follow Naor's threshold
strategy and never renege. B-customers instead
use a threshold strategy that depends on the
number of customers of each class ahead of them.
We show that this threshold is weakly decreasing
in the number of B-customers ahead and derive
a recursive characterization using one-dimensional
ruin probabilities and expected absorption times.
In particular, the equilibrium strategy of
B-customers does not depend on their arrival rate.

We then consider the case of identical service
rates. We show that the equilibrium threshold
reduces to a truncated linear function of the
number of B-customers ahead, so that the
equilibrium decision depends only on the total
number of customers ahead. This recovers the
equilibrium structure characterized by Afèche
and Sarhangian (2026).

Finally, we investigate the behavior of the
equilibrium when service rates differ. We prove
local stability around identical service rates
under a non-indifference condition and show
that linear approximations may have arbitrarily
large uniform errors, even when their coefficients
are chosen optimally. We also derive explicit
linear upper and lower bounds, which provide
sufficient conditions for an accurate linear
approximation.

%
%

\subsection{Related Literature}
\label{suse:related-literature}

In his seminal paper \citet{Nao:E1969} introduced the first model of strategic queueing theory. 
In his \MMone \ac{FCFS} model, upon arrival at a queue, customers  decide whether to join it or to balk.
They make their decision based on the parameters of the \MMone model, on the reward they receive when served and on the cost they pay for each unit of time they stay in the system. 
This defines a dynamic game.
Naor proved that, if customers make their decisions according to a selfish criterion, \ie if they play a Nash equilibrium of the above game, the outcome is socially suboptimal, that is, a social planner could achieve a better outcome.
\citet{Knu:E1972} generalized Naor’s result to a multi-server system with non-linear waiting cost functions.
More generally, \citet{HasSni:OR2020} studied individual, social, and monopoly optimization in observable queues, and provided general conditions under which the usual ordering of individual, social, and monopoly thresholds extends beyond Naor's original model.

\citet{Has:E1985} proved that efficiency can be achieved at equilibrium by using a different regime, \eg \ac{LCFS}.
Other regimes that achieve social optimality in equilibrium were studied, among others, by \citet{Wan:MS2016,HavOz:ORL2016,HavOz:QMSM2018,CheTer:JPE2026,ScaShm:EC2024,ScaShm:arXiv2024}.
Related issues concerning learning and  strategic experimentation were considered, among others, by \citet{CriTho:TE2019,Mar:TE2025}.
In all these papers customers are assumed to be homogeneous and no priority is allowed.

There exists a sizeable literature that considers  queueing systems with priorities in a non-strategic setting.
The first papers in this stream of literature go back to the fifties.
Initially, several papers studied the behavior of various queueing systems with priority classes,  under different  protocols, for instance preemption versus no-preemption or resume versus no-resume for the preempted customers.
Among the various contributions to this literature, we mention 
\citet{Cob:JORSA1954,Cob:JORSA1955},
\citet{Gav:JRSSB1962,Gav:JRSSB1963},
\citet{Hol:JORSA1954},
\citet{KesRun:IM1957a,KesRun:IM1957b},
\citet{WhiChr:OR1958},
\citet{Ste:OR1958},
\citet{Mil:AMS1960}, \citet[chapter~3.3]{CoxSmi:Wiley1961},
\citet{Avi-Nao:OR1961,Avi:OR1963}.


Some papers studied optimal tolls with the goal to achieve optimality in a system where customers act strategically and selfishly aim at maximizing their own payoff. 
For instance, \citet{Lit:IJPR1974} studied the optimum pricing policy for an \MMone \ac{FCFS} model with heterogeneous customers.

Other papers compared preemptive and non-preemptive priority in M/G/1 queues. 
Among them, \citet{AviBroNao:ZAMM1964} considered a policy that is intermediate between the head of the line priority, and the preemptive priority, \ie a policy that preempts a low-priority customer who has not spent much time being served, but does not affect a customer who has been under service for a long time.
The paper optimizes over the preemptive threshold. 
A similar approach was taken by \citet{Bal:OR1970} who studied the same priority rules under various assumptions on the service distribution.
These results were further examined by \citet{Bal:MS1971} under a linear cost model.

An entire stream of literature considered the optimal design of queueing systems with priority classes and non-strategic customers. 
Among them, \citet{CofMit:OR1980}, \citet{BroWal:QS2006} and \citet{WuAyh:QS2025}.
The reader is referred to 
\citet{Sti:OR2002,Sti:CRC2009} for a survey and an extended treatment of this literature.

In a strategic setting, \citet{Bal:MS1972} studied a GI/G/1 \ac{FCFS} model where customers, when they arrive at the queue, based on some information, can pay a fee to gain priority of service. 
His analysis was extended by \citet{TilBal:EJOR1979}, who considered various queueing regimes where customers can decide whether to balk or join the queue and can pay to get priority.
\citet{BalSch:E1979,BalSch:EJOR1980} proved that in a multi-class  M/GI/1 \ac{FCFS}  model without reneging, with an unobservable queue, and where the arrival rates can be chosen by the customers or the social planner, the class that dominates the facility in the equilibrium is not necessarily the same one that dominates it in the social optimum. 
Their results were extended by \citet{BalRad:MS1994} who considered also \ac{LCFS} regimes with and without priority, and with pricing.










\citet{AviMaxMil:OR1965} studied a non-strategic queueing system with two classes of customers, where priorities can switch from one class to the other. 
In particular, when the system is empty, the first arriving customer obtains priority for their class, and this priority is maintained as long as at least one customer of that class remains in the system. 
Once all customers of that class have been served, the priority passes to the customers of the other class (if any are in the queue) or to the class of the first new arrival if the system is empty. 
In this setup, the paper analyzed the expected waiting time and queue length using a special mathematical tool to characterize the steady-state densities.

Several papers studied queueing models with priority, where customers can pay to achieve a higher priority.
\citet{AdiYec:OR1974} considered a model where there are different priority classes and customers can pay a fee to be a member of one of these classes.  
This model was re-examined by 
\citet{HasHav:OR1997}, who showed that multiplicity of equilibria can occur.
\citet{MenWan:OR1990}
examined a model with multiple user classes and provided a pricing mechanism that is optimal and incentive compatible. 
\citet{AfeMen:MS2004} studied systems where customers can bid for service priorities under both preemptive and non-preemptive regimes. They propose a price-service mechanism with a generalized delay-cost structure that extends the standard additive model by incorporating a multiplicative component. This interdependence between service value and waiting times can reverse classical results on revenue versus social optimality.
\citet{HasHav:ORL2006} considered an \MMone model with two classes of customers, where the queue is unobservable and customers pay a class-dependent fee to enter the system. 
The server knows the customers' classes and can set the fees.
\citet{HavRav:AOR2016} studied a strategic bidding system where customers bid for priority in an M/G/1 queue.
\citet{AfePav:MS2016} analyzed how a firm can design revenue-maximizing price/lead-time menus and scheduling policies for heterogeneous, time-sensitive customers with private information. 
\citet{AfeBarMilRoe:OR2019} showed that, in queueing systems where customers differ in their demand rates, it may be optimal to offer delay-differentiated services through priority mechanisms rather than uniform service disciplines such as FIFO, even when all customers are equally sensitive to delays. Importantly, this result holds under both aggregate and individual control models, suggesting that in some cases firms may benefit from prioritizing customer classes with lower marginal waiting costs.

\citet{CheKul:Springer2006,CheKul:SM2007} studied an \MMone queueing  model with two priority classes where customers can decide whether to join a queue or balk. 
Reneging is not allowed. 
They compare the equilibrium outcome, the social optimum, and an intermediate situation where each priority class has a social planner who cares only about the welfare of customers in that class. 
A similar model in the M/G/1 framework was considered by \citet{XuXuWan:JAMC2016,XuXuYao:JIMO2019}.

\citet{EngHas:QS2017} considered a model with two \ac{FCFS} queues, a system queue and a virtual queue, where the virtual queue has a lower waiting cost, but the service queue has non-preemptive priority over the virtual queue. 
Customers decide which queue to join. 
The paper studied the equilibrium of this game under two different information structures.
Related models were studied by \citet{ArmMag:OR2024a,ArmMag:OR2024b}.

The model in  \citet{AfeSar:MSfrth} is the closest to the one in our paper. 
They considered an \MMone \ac{FCFS} queueing regime with two priority classes of customers.
In the first part of their paper, they studied the customers' equilibrium behavior, focusing on low-priority customers, given that high-priority customers use a threshold strategy, where the threshold is as in \citet{Nao:E1969}.
Differently from what we do, in their model the service rates of the two classes are the same. 
This allowed them to  parametrize the state in terms of the total number of customers and the number of low-priority customers in the system. 
This parametrization is different from ours, but totally equivalent when the service rates of the two classes coincide.
In the second part of their paper, these authors  studied the role of pricing to control customers' balking and reneging.

Another related stream of literature studied two-sided, or double-ended, queues with strategic agents. 
In these models, agents from two populations arrive on opposite sides of a matching system and decide whether to join or balk. 
\citet{NguPhu:ORL2022} studied an observable double-ended queue with two strategic populations and zero matching times, showing that the game may admit multiple Nash equilibria but a unique subgame-perfect Nash equilibrium. 
More recently, \citet{AfeAkaBenDabDen:SSRN2025} revisited  Naor's model in a two-sided matching environment, characterizing equilibrium joining thresholds and comparing self-optimization, social welfare maximization, and profit maximization.

Health care is an area where strategic queues with priorities are studied extensively. 
In an \acdef{ED} after triage, patients are admitted according to the severity of their symptoms. 
Patients can decide not to join the queue or to abandon it if their conditions are not serious and the queue ahead of them is or becomes too long. 
\citet{BatTer:MS2015} provided an empirical analysis of abandonment in \acp{ED} depending on their information, which is typically the length of the queue, but not the priority of the patients in the queue.
A survey of queueing problems in \acp{ED} is provided by \citet{ElaWac:ORF2021}.
In particular, their section~5.2 deals with priority models.

Several papers have studied the case where customers can purchase priority. 
Among them, \citet{HasHav:OR1997}

The reader is referred to \citet{HasHav:Kluwer2003} for an extensive treatment of strategic queueing and to 
\citet{Has:CRC2016} for a detailed review of the more recent relevant contributions in the field.

Some papers considered classical non-strategic queueing models with priority classes. 
For instance, \citet{WanBarSch:OR2015} analyzed an \MMone[s] queueing model with two priority classes in a non-strategic framework. 
Their analysis requires the reduction of a bi-dimensional Markov chain to a uni-dimensional Markov chain. 
This is similar to what we have in our paper.
In the same framework, 
\citet{RasIngSan:IJC2022} developed algorithms for multiserver queueing systems with two customer classes, preemptive priorities, and reneging. 
To do this they used a quasi-birth-death process.
\citet{XuSha:OR1993} studied a \ac{FCFS} $\MMone[s]$ queueing system with the objective of maximizing the expected discounted social utility. 
For our purposes, the main relevance of their work lies in their methodological approach: they introduced a dual \ac{LCFSP} system, proved that it is isomorphic to the original \ac{FCFS} system, and showed that the two systems admit the same optimal policy. 
This duality simplifies the analysis of the original problem.

%
%
  
\subsection{Organization of the Paper}
\label{suse:organization}

The general model is described in \cref{se:model}.
\cref{se:equilibrium} characterizes the equilibrium strategies and
provides a recursive construction of the equilibrium threshold.
The particular case where the two classes have the same service rate is considered in \cref{se:equal-rates}.
\cref{se:unequal-service-rates} establishes
local stability, examines the limitations of
linear approximation, and derives explicit
linear bounds for the equilibrium threshold.

%
%

\subsection{Notation}
\label{suse:notation}
Given a real number $x$, the symbol $\infloor{x}$ indicates the greatest integer $n$ such that  $n\le x$. 
The notation $X\sim \distr$ means that the random variable $X$ is distributed according to the distribution function $\distr$. 




%
%
    
\section{The Model}
\label{se:model}

Before dealing with our priority-class model, we recall the classical contributions on strategic balking and reneging in an \MMone queueing model with homogeneous customers.

%
%

\subsection{Homogeneous Customers}
\label{suse:single-class}

\citet{Nao:E1969} studied an \MMone \ac{FCFS} queueing model with waiting costs and reward for service.
More specifically:

\begin{itemize}
\item 
Customers arrive according to a Poisson process with parameter $\rate$.

\item 
Service times are \ac{iid} according to an exponential distribution with parameter $\rateS$.

\item 
Each customer receives a reward $\reward$ upon service completion and incurs a cost $\Cost$ for each unit of time spent in the system.

\item 
Upon arrival at the queue, each customer decides whether to join it or to balk.

\item 
A customer (Chantal) who joins the queue receives a  payoff 
\begin{equation}
\label{eq:stoch-time}
\reward - \Cost \Wtime,    
\end{equation}
where $\Wtime$ is the random time she spends in the system.

\item 
A customer who balks receives a payoff equal to zero.

\item 
A customer who balks will never come back to the queue.

\end{itemize}

Customer Chantal joins a queue with $\nAB$ customers if and only if her expected payoff is nonnegative, \ie
\begin{equation}
\label{eq:expect-time}
\reward - (\nAB+1)\frac{\Cost}{\rateS} \ge 0.    
\end{equation}
As a consequence, the integer 
\begin{equation}
\label{eq:Naor-threshold} 
\eqn \coloneqq \floor*{\frac{\reward\rateS}{\Cost}}
\end{equation}
is such that
\begin{equation}
\label{eq:expect-time-equil}
\reward - \eqn\frac{\Cost}{\rateS} \ge 0
\quad\text{and}\quad
\reward - (\eqn+1)\frac{\Cost}{\rateS} < 0.    
\end{equation}
Therefore, the threshold strategy in which Chantal balks if she observes $\eqn$ or more customers in the queue is an equilibrium strategy.
Once she enters the queue, Chantal will never have an incentive to renege.

\subsection{Priority Model}
\label{subsec:PriorityModel}

We now consider an observable \MMone queueing system with two classes of customers, denoted by $\typeA$ and $\typeB$. 
Customers of class $\typeA$ have preemptive priority over customers of class $\typeB$, so $\typeB$-customers are served only when no $\typeA$-customers are in the system. 
Within each class, customers are served according to a first-come first-served discipline. 
Furthermore, an arriving $\typeA$-customer overtakes all $\typeB$-customers, including a $\typeB$-customer who is currently in service.

\begin{definition}
\label{de:state}   
For any  $\ttime\in\reals_{+}$, we say that customer Chantal is \emph{active} at time $\ttime$ if  she either is in the system or arrives at the queue at time $\ttime$. 
We let $\presplayt$ denote the set of \emph{active customers} at time $\ttime$.

At any time $\ttime$, the \emph{state of any customer $\cust\in\presplayt$} is the pair $(\nAi,\nBi)\eqqcolon\state_{\cust}$ where, for $\type\in\braces{\typeA,\typeB}$, $\ntypei$ is the number of class-$\type$ customers currently ahead of $\cust$ in the service order. 
If $\cust$ is an arriving customer, her
state is the state she would have upon joining the system.

The \emph{state of the system} at time $\ttime$ is the tuple $\state\coloneqq\parens*{\state_{\cust}, \cust\in\presplayt}.$
\end{definition}

Once the priority is established, within each class the system follows a \ac{FCFS} regime. 
As previously mentioned, upon arrival each customer can observe the current state of the system.
Each customer of class $\type\in\{\typeA,\typeB\}$ incurs  a cost $\Costtype$ for each unit of time spent in the system, and receives a reward $\rewardtype$ upon service completion. 
The model can be formalized as follows:

\begin{enumerate}[label={\rm (\roman*)}, ref=(\roman*)]
\item A stream of customers $\typeA$ arrives according to a Poisson process with parameter $\rateA$.
    
\item 
A stream of customers $\typeB$ arrives according to a Poisson process with parameter $\rateB$.

\item 
The two Poisson processes are independent.

\item 
The service times of  all customers are independent random variables. 

\item 
For $\type\in\braces*{\typeA,\typeB}$, the service time of a $\type$ -customer has an exponential distribution with parameter $\rateS_{\type}$.
\end{enumerate}

\begin{center}
\begin{tikzpicture}[>=stealth]

\node[circle, draw, inner sep=2pt] (B0) at (0,0) {$\typeB$} ;
\draw[->] (0,-0.5) -- ++(0,-1.5);
\filldraw (0,-1.2) node[anchor=east] {$\Poisson(\rateB)$};

\node[circle, draw, inner sep=2pt] (A0) at (3,0) {$\typeA$};
\draw[->] (3,-0.5) -- ++(0,-1.5);
\filldraw (3,-1.2) node[anchor=east] {$\Poisson(\rateA)$};

\foreach \x in {0,...,2}
    \node[circle, draw, inner sep=2pt] (B\x) at (\x,-2.5) {$\typeB$};

\foreach \x/\label in {3/A, 4/A, 5/A}
    \node[circle, draw, inner sep=2pt] (A\x) at (\x,-2.5) {$\typeA$};

\draw[->] (5.5,-2.5) -- (7.5,-2.5);
\filldraw (6.5,-1.5) node[anchor=north] {$\Exp(\rateStype)$};

\node[rectangle, draw, inner sep=2pt] (A0) at (8.5,-2.5) {Server};
\end{tikzpicture}
\end{center}

Moreover, for  $\type\in\braces*{\typeA,\typeB}$, let $\utilizationtype\coloneqq\ratetype/\rateStype$ denotes the utilization factor for $\type$-customers. 
Notice that, if $\rewardtype\rateStype < \Costtype$, then customers of type $\type$ have an incentive to never join the system.
That is, even if they join a system with no other customer ahead of them, the expected cost due to the service time is already higher than the reward.

All parameters in the model are common knowledge and the queue is observable. 
An arriving customer decides whether to join the queue or balk. 
Moreover, a customer in the queue can  renege, before getting served.
However, only $\typeB$-customers may have a reason to renege because, unlike $\typeA$-customers, their position in the queue can get worse with the arrival of higher-priority customers. 

Whenever an exogenous event changes the queue configuration, potential immediate reneging decisions of $\typeB$-customers are resolved in reverse order in the $\typeB$-queue, starting from the last customer and proceeding toward the front of the queue. 
A customer who reneges is removed before the decision of the next customer ahead is made. 
Because every queue configuration contains finitely many customers, any such sequence of instantaneous reneging decisions terminates after finitely many steps.

Throughout the paper we consider a representative $\typeB$-customer $\cust$ and call him Bernard.

%
%

\section{Equilibrium}
\label{se:equilibrium}

Consider the priority queueing regime described in \cref{subsec:PriorityModel}.
This represents a game, which we call \emph{priority queueing game}.

\begin{definition}
\label{de:MS-MPE}   A \emph{Markov strategy} of customer $\cust$ is a function $\strat^{(\cust)}$ that associates to each state $(\nAi,\nBi)$ a decision to either stay or abandon the system. 

A Markov strategy profile $\strat\coloneqq\parens*{\strat^{(\cust)}}_{\cust\in\naturals}$ is a \acdef{MPE} if, for every state $\state$, the profile is a Nash equilibrium in the game that starts in state $\state$, in which customers can decide whether to stay in the system or to renege. 
\end{definition}

We refer to an \ac{MPE} simply as an equilibrium. 
Knowing the state $(\nAi,\nBi)$, a newly arriving customer $\cust$ decides to join the queue if and only if $\cust$'s expected future payoff is nonnegative. 

Since $\typeA$-customers have priority over $\typeB$-customers, their behavior is not influenced by the arrival of $\typeB$-customers. 
Therefore, the equilibrium behavior of $\typeA$-customers is the same as in \citet{Nao:E1969}, \ie  when arriving at the queue, an $\typeA$-customer joins if and only if $\nA<\naor$, where 
\begin{equation}
\label{eq:Naor}
\naor=\floor*{\frac{\rewardA\rateSA}{\CostA}}.
\end{equation}

We now turn to Bernard's  equilibrium behavior. 
The next result shows that, in any symmetric MPE in stationary strategies, Bernard's equilibrium decision has a threshold structure in the number of $\typeA$-customers ahead, and that the corresponding
threshold is weakly decreasing in the number of $\typeB$-customers ahead.

Throughout the paper, whenever a customer is indifferent, we assume that she joins or remains in the system.

\begin{theorem}
\label{thm:BThrsEq} 
In any symmetric MPE in stationary strategies,  $\typeA$-customers adopt the threshold strategy described in \cref{eq:Naor} and there exists a weakly decreasing function $\eqthreshB:\naturals \to \{-1,0,\dots,\naor\}$ such that each $\typeB$-customer $\cust$ in state $(\nAi,\nBi)$ joins or remains in the system if and only if $\nAi\leq\eqthreshB(\nBi)$. 
\end{theorem}

\begin{proof}
Fix a symmetric MPE in stationary strategies. 
Since the behavior of $\typeA$-customers is independent of that of $\typeB$-customers, in any equilibrium $\typeA$-customers follow the Naor threshold strategy with threshold $\naor$.

Let's now focus on $\typeB$-customers. 
Let $\bellman(\nAi,\nBi)$ denote Bernard's continuation payoff from remaining in the system in the current state and behaving optimally thereafter, given that all other customers follow their equilibrium
strategies.
Bernard remains in the system in state $(\nAi,\nBi)$ if and only if $\bellman(\nAi,\nBi)\geq 0$.

We first establish the monotonicity properties of Bernard's continuation payoff. 
Fix $\nAi$ and consider two systems in which Bernard's state is either $(\nAi,\nBi)$ or $(\nAi,\nBi+1)$.
Construct the second system from the first by placing one additional $\typeB$-customer (Belinda) immediately ahead of Bernard, and couple the two systems using the same arrival and service processes. 
The $\nBi$ class-$\typeB$ customers that are ahead of Bernard in both systems have the same individual states; therefore, by symmetry and stationarity of the equilibrium, they take the same actions under the coupling. 
Belinda's presence  in the second system cannot improve Bernard's position: if she reneges, the two systems coincide from that time onward, whereas, if she remains, she can only delay Bernard's service. Therefore,
\begin{equation}
\label{eq:ContPayMonB}
\bellman(\nAi,\nBi+1)\leq \bellman(\nAi,\nBi).
\end{equation}
Because all $\typeB$-customers adopt the same stationary equilibrium strategy, a direct consequence of \cref{eq:ContPayMonB} is that, as long as Bernard finds it optimal to remain in the system, no $\typeB$-customer ahead of him reneges. 
Therefore, while Bernard stays in the system, the number of $\typeB$-customers ahead can decrease only through a $\typeB$-service completion.

Next, we establish monotonicity with respect to $\nAi$.
Fix $\nBi$ and consider two systems where Bernard's state is either $(\nAi,\nBi)$ or $(\nAi+1,\nBi)$.
As shown above, as long as Bernard remains in the system, no $\typeB$-customer ahead of him reneges. 
Hence, for as long as Bernard remains in the system, the state evolves only through $\typeA$-arrivals, $\typeA$-service completions, and $\typeB$-service completions.
Couple the two systems using the same $\typeA$-arrival, $\typeA$-service, and $\typeB$-service processes.
Starting from $(\nAi,\nBi)$ and $(\nAi+1,\nBi)$, respectively, the number of $\typeA$-customers in the first system is never larger than in the second. 
Indeed, $\typeA$-arrivals and $\typeA$-service completions preserve the order, and the first system becomes available for class-$\typeB$ service no later than the second.
Therefore, Bernard's position in the first system is always weakly better than in the second, which implies that, under the coupling, in the first system Bernard cannot be served later than in the second. 
Because he incurs the same holding-cost rate $\CostB$ in both systems and receives the same reward $\rewardB$ upon service, his optimal continuation payoff is weakly larger in the first system. 
Therefore,
\begin{equation}
\label{eq:ContPayMonA}
\bellman(\nAi+1,\nBi)\leq \bellman(\nAi,\nBi) \quad\text{for } \nAi=0,\dots,\naor-1.
\end{equation}
For each $\nBi\geq 0$, let
\begin{equation}
\label{eq:A-nB}    
\AnB \coloneqq
\braces*{0\leq\nAi\leq \naor: \bellman(\nAi,\nBi)\geq0}.
\end{equation}
By \eqref{eq:ContPayMonA}, if $\nAi\in \AnB$, and $\nAi{}'<\nAi$, then $\nAi{}'\in\AnB$. 
Hence,
\begin{equation}
\label{eq:A-nB=}  
\AnB=\{0,\dots, \eqthreshB(\nBi)\}
\end{equation}
for some $\eqthreshB(\nBi)\in \braces*{-1,0,\dots,\naor}$, where $\eqthreshB(\nBi)=-1$, if $\AnB=\varnothing$. 
Therefore, Bernard remains in the system if and only if  $\nAi\leq \eqthreshB(\nBi)$.

Finally, \cref{eq:ContPayMonB} implies $\Acal (\nBi+1)\subseteq \AnB$.
Hence,
\begin{equation}
\label{eq:K-monotone}  
\eqthreshB(\nBi+1)\leq \eqthreshB(\nBi),
\end{equation}
\ie $\eqthreshB$ is weakly decreasing, which completes the proof.

\end{proof}

\begin{remark}
\label{rem:equationBellman}
For later reference, we state  the Bellman equations satisfied by Bernard's continuation payoff.
Define
\begin{equation}
\label{eq:V-max-U}  
\valueB(\nAi,\nBi) \coloneqq \max\{0,\bellman(\nAi,\nBi)\}.
\end{equation}
As shown in the proof of \cref{thm:BThrsEq}, as long as Bernard remains in the system, no $\typeB$-customers ahead of him renege. 
Moreover, future $\typeB$-arrivals do not affect his state.
Hence,
\begin{align}
\label{eq:Bellman-U}
\bellman(\nAi,\nBi)
&=
\begin{dcases}
\frac{-\CostB+\rateA\valueB(\nAi+1,\nBi)+\rateSA\valueB(\nAi-1,\nBi)}{\rateA+\rateSA}, & 0<\nAi<\naor,\\
-\frac{\CostB}{\rateSA}+\valueB(\naor-1,\nBi), & \nAi=\naor,\\
\frac{-\CostB+\rateA\valueB(1,\nBi)+\rateSB\valueB(0,\nBi-1)}{\rateA+\rateSB}, &\nAi=0, \nBi>0,
\end{dcases}
\intertext{with}
\label{eq:Bellman-U-0}
\bellman(0,0)
&=\frac{-\CostB+\rateA\valueB(1,0)+\rateSB \rewardB}{\rateA+\rateSB}.
\end{align}
\end{remark}

We now characterize the equilibrium-threshold function $\eqthreshB(\argdot)$. 
We begin with $\eqthreshB(0)$.  
When $\nB^{(\cust)}=0$, no $\typeB$-customer is ahead of Bernard and all future $\typeB$-arrivals join behind him. 
Hence, once the strategy of $\typeA$-customers is fixed, Bernard's payoff depends only on his own stopping strategy.

For $\threshgen\in\{0,\ldots,\naor\}$, let $\stratk$ denote the strategy under which, in state $(\nAi,0)$, Bernard remains in the system if and only if $\nAi\le\threshgen$, when all $\typeA$-customers use the $\naor$-threshold strategy.

When $\typeA$-customers follow the threshold strategy in
\cref{eq:Naor}, let
$\Prewardk(\nAi,0)$ denote the probability that Bernard is served and let $\Etimek(\nAi,0)$ denote the expected total time that Bernard spends in the system under $\stratk$.

When $\valb/\CostB\geq 1/\rateSB$, call $\mAstar(\valb)$ the unique nonnegative solution $\nAstar$ of the following equation:
\begin{equation}
\label{eq:m-A}
\frac{\valb}{\CostB}\rateSB
=
\begin{dcases}
\begin{aligned}
&\frac{1}{(\rateSA-\rateA)}\left(\frac{\nAstar(1-\utilizationA^{\nAstar+1})- (\nAstar+1)\utilizationA(1-\utilizationA^{\nAstar})}{1-\utilizationA^{\nAstar+1}}\right.\times\\
&\quad\frac{\rateSB(1-\utilizationA^{\nAstar+1})+\rateA\utilizationA^{\nAstar}(1-\utilizationA)}{(1-\utilizationA)}
\\
&\left.\qquad+\rateSA\frac{1-\utilizationA^{\nAstar+1}-(\nAstar+1)\utilizationA^{\nAstar+1}(1-\utilizationA)}{1-\utilizationA^{\nAstar+1}}\right)
\end{aligned}
& \text{for }\utilizationA\neq 1,\\
\displaystyle{\frac{(\nAstar+1)(\nAstar\rateSB+2\rateSA)}{2\rateSA}}, &\text{for }\utilizationA=1.
\end{dcases}
\end{equation}

\begin{lemma}
\label{lem:CharcK0}
Let $\eqthreshB$ be the threshold function defined in \cref{thm:BThrsEq}.
Then
\begin{equation}
\eqthreshB(0)=
\begin{dcases}
-1, & \text{if }  \frac{\rewardB}{\CostB}<\frac{1}{\rateSB},\\
\floor{\mAstar(\rewardB)},  & \text{if }\frac{1}{\rateSB}\leq\rewardB/\CostB<\Etime^{\naor}(\naor,0)\\
\naor, & \text{if }\rewardB/\CostB\geq \Etime^{\naor}(\naor,0).
\end{dcases}
\end{equation}
\end{lemma}

\begin{proof}
As mentioned before, if $\rewardB/\CostB<1/\rateSB$, then it is not profitable for Bernard to be in the system even in state $(0,0)$. 

Suppose that $\rewardB/\CostB \geq 1/\rateSB$ and assume that Bernard's state is $(\nAi,0)$.
Because the number of $\typeA$-customers cannot exceed $\naor$, and any future $\typeB$-customer who joins will be behind him, Bernard's state can never become worse than $(\naor,0)$. 
Therefore, if Bernard remains in the system in state $(\naor,0)$, he will stay until service completion. 
Bernard remains in the system in state $(\naor,0)$ if and only if
\begin{equation}
\label{eq:FSNetBenefitIneq}
\rewardB-\CostB\Etime^{\naor}(\naor,0)\geq0.
\end{equation}
It follows that, if $\rewardB/\CostB\geq \Etime^{\naor}(\naor,0)$, then $\eqthreshB(0)=\naor$.

Suppose now that $\rewardB/\CostB< \Etime^{\naor}(\naor,0)$. 
Then $\eqthreshB(0)<\naor$. 
To determine its value, consider a candidate threshold $\threshgen\in\braces*{0,\ldots,\naor-1}$, and let $\stratk$ be  Bernard's stopping strategy.

Starting from the boundary state $(\threshgen,0)$, the expected payoff obtained under $\stratk$ is
\begin{equation}
\label{eq:UtilityIneq1}
\rewardB\Prewardk(\threshgen, 0)-\CostB\Etimek(\threshgen, 0).
\end{equation}
Since $\eqthreshB$ is a threshold strategy, it follows that
\begin{equation}
\label{eq:K(0)}  
\eqthreshB(0)=\max\{\threshgen\in\{0,\dots,\naor-1\} \colon \rewardB\Prewardk(\threshgen, 0)-\CostB\Etimek(\threshgen, 0)\geq0 \}.
\end{equation}

We first compute the probability $\Prewardk(\threshgen, 0)$ that Bernard is served. 
To do this, we need to compute the conditional probability $\Pruin_{\initnum}^{\threshgen+1}$ that the number of $\typeA$-customers becomes $0$ before becoming $\threshgen+1$, given that the initial number of $\typeA$-customers is $\initnum$.
We can interpret $\Pruin_{\initnum}^{\threshgen+1}$ as the ruin probability in a random walk where the initial position is $\initnum$, the goal is $\threshgen+1$, and the winning probability at each round equals $\rateA/(\rateSA+\rateA)$. 
The value $\threshgen+1$ will remain  fixed throughout; therefore, for simplicity, we write $\Pruin_{\initnum}$ instead of $\Pruin_{\initnum}^{\threshgen+1}$.

By \cref{theo:GamblerRuin1}, for every $\initnum=0,\dots,\threshgen+1$,
\begin{equation}
\label{eq:RuinProba}
\Pruin_{\initnum}=
\begin{dcases}
\frac{1-\utilizationA^{\threshgen+1-\initnum}}{1-\utilizationA^{\threshgen+1}} & \text{ if } \utilizationA\neq 1 \\ \frac{\threshgen+1-\initnum}{\threshgen+1} & \text{ if } \utilizationA=1.
\end{dcases}
\end{equation}
Using conditional probability, we can express $\Prewardk(\threshgen, 0)$ as the product of two terms: the probability $\Pruin_{\threshgen}$ that all the $\typeA$-customers ahead of Bernard get served before becoming $\threshgen+1$, and the probability $\Prewardk(0, 0)$ that Bernard gets served when he is alone in the system: 
\begin{equation}
\label{eq:Preward0}
\Prewardk(\threshgen, 0)= \Pruin_{\threshgen}\cdot \Prewardk(0, 0).
\end{equation}
Conditioning on the first step of the random walk, we obtain
\begin{align}
\Prewardk(0, 0)&= \frac{\rateSB}{\rateA+\rateSB}+ \frac{\rateA}{\rateA+\rateSB} \Prewardk(1,0) \\
&= \frac{\rateSB}{\rateA+\rateSB}+ \frac{\rateA}{\rateA+\rateSB} \Pruin_{1}\cdot\Prewardk(0,0). 
\end{align}
Therefore,
\begin{align}
&\Prewardk(0,0)\parens*{1-\frac{\rateA}{\rateA+\rateSB} \Pruin_{1}}
=\frac{\rateSB}{\rateA+\rateSB},
\intertext{that is,}
&\Prewardk(0,0) 
=\frac{\rateSB}{\rateSB+\rateA\left(1-\Pruin_{1}\right)}.
\end{align}    
Using \cref{eq:RuinProba}, we obtain 
\begin{equation}
\label{eq:Preward1}
\Prewardk(0, 0)=\begin{dcases}
\frac{\rateSB\left(1-\utilizationA^{\threshgen+1}\right)}{\rateSB\left(1-\utilizationA^{\threshgen+1}\right)+\rateA\utilizationA^{\threshgen}(1-\utilizationA)} & \text{ if } \utilizationA\neq 1, \\ \frac{\rateSB(\threshgen+1)}{\rateSB(\threshgen+1)+\rateA} & \text{ if } \utilizationA=1.
\end{dcases}
\end{equation}
Combining \cref{eq:RuinProba,eq:Preward1} we see that \cref{eq:Preward0} becomes
\begin{equation}
\label{eq:Preward2}
\Prewardk(\threshgen, 0)
=\begin{dcases}          \frac{\rateSB\left(1-\utilizationA\right)}{\rateSB\left(1-\utilizationA^{\threshgen+1}\right)+\rateA\utilizationA^{\threshgen}(1-\utilizationA)} & \text{ if } \utilizationA\neq 1, \\ \frac{\rateSB}{\rateSB(\threshgen+1)+\rateA} & \text{ if } \utilizationA=1.
\end{dcases}
\end{equation}
We now focus on $\Etimek(\threshgen, 0)$, the total expected time spent by Bernard in the system. 
Condition on the event that the number of $\typeA$-customers in the system is $\initnum$, and define $\Eruin_{\initnum}^{\threshgen+1}$ the expected time until the number of $\typeA$-customers becomes either 0 or $\threshgen+1$. 
We can interpret $\Eruin_{\initnum}^{\threshgen+1}$ as the expected duration of a random walk  that starts at $\initnum$ and stops either at $0$ or at $\threshgen+1$, where the winning probability at each round equals $\rateA/(\rateSA+\rateA)$. 
As before, the value $\threshgen+1$ will be  fixed throughout, so, for the sake of simplicity, we write $\Eruin_{\initnum}$ instead of $\Eruin_{\initnum}^{\threshgen+1}$.

Before reaching either $0$ or $\threshgen+1$, each transition occurs after an
exponential time with parameter $\rateA+\rateSA$. Hence, combining
the expected number of transitions from \cref{theo:GamblerRuin2} with the mean
inter-transition time $1/(\rateA+\rateSA)$, for every $\initnum=0,\dots,\threshgen+1$, we have
\begin{equation}
\label{eq:RuinEtime}
\Eruin_{\initnum}
=\begin{dcases}
\frac{1}{\rateSA-\rateA}\parens*{\initnum-\frac{(\threshgen+1)\utilizationA^{\threshgen+1-\initnum}(1-\utilizationA^{\initnum})}{1-\utilizationA^{\threshgen+1}}} 
& \text{ if } \utilizationA\neq 1, \\ \frac{\initnum(\threshgen+1-\initnum)}{2\rateSA} 
& \text{ if } \utilizationA=1.
\end{dcases}    \end{equation}
We can decompose $\Etimek(\threshgen, 0)$ as the sum of two terms: the expected time $\Eruin_{\threshgen}$ until Bernard  either renege or becomes the only customer in the system plus the probability of becoming the only customer in the system times the conditional expected time that he has to spend in the system, when this happens:
\begin{equation}
\label{eq:Etime0}
\Etimek(\threshgen, 0)=\Eruin_{\threshgen} + \Pruin_{\threshgen}\Etimek(0, 0),
\end{equation}
Again, conditioning on the first step of the random walk, we obtain
\begin{align}
\Etimek(0, 0)&= \frac{\rateSB}{\rateA+\rateSB}\frac{1}{\rateA+\rateSB}+ \frac{\rateA}{\rateA+\rateSB}\left(\frac{1}{\rateA+\rateSB} +\Etimek(1,0)\right) \nonumber\\
&= \frac{1}{\rateA+\rateSB}+ \frac{\rateA}{\rateA+\rateSB}\left( \Eruin_{1} + \Pruin_{1}\Etimek(0, 0)\right). \nonumber
\end{align}
Therefore,
\begin{align}
&\Etimek(0,0)\left(1-\frac{\rateA}{\rateA+\rateSB} \Pruin_{1}\right)
=\frac{1}{\rateA+\rateSB} + \frac{\rateA}{\rateA+\rateSB} \Eruin_{1}.\\
\intertext{Hence,}
&\Etimek(0,0) 
=\frac{1}{\rateSB+\rateA\left(1-\Pruin_{1}\right)}+\frac{\rateA \Eruin_{1}}{\rateSB+\rateA\left(1-\Pruin_{1}\right)}.
\end{align}
Using \cref{eq:RuinEtime} and \cref{eq:RuinProba}, we obtain
\begin{equation}
\label{eq:Etime1}
\Etimek(0, 0)
=\begin{dcases}
\frac{\rateSA}{\rateSA-\rateA}\frac{1-\utilizationA^{\threshgen+1}-(\threshgen+1)\utilizationA^{\threshgen+1}(1-\utilizationA)}{\rateSB\left(1-\utilizationA^{\threshgen+1}\right)+\rateA\utilizationA^{\threshgen}(1-\utilizationA)} 
& \text{ if } \utilizationA\neq 1 \\ \frac{1}{2}\frac{(\threshgen+2)(\threshgen+1)}{\rateSB(\threshgen+1)+\rateA} 
& \text{ if } \utilizationA=1.
\end{dcases}
\end{equation}
Combining \cref{eq:RuinEtime,eq:Etime1}, we see that \cref{eq:Etime0} becomes
\begin{equation}
\label{eq:Etime2}
\Etimek(\threshgen, 0)
=\begin{dcases}
\frac{1}{\rateSA-\rateA}\left(\frac{\threshgen(1-\utilizationA^{\threshgen+1})- (\threshgen+1)\utilizationA(1-\utilizationA^{\threshgen})}{1-\utilizationA^{\threshgen+1}}
\right. & \\
\quad+\left.\frac{\rateSA(1-\utilizationA)}{1-\utilizationA^{\threshgen+1}}\frac{1-\utilizationA^{\threshgen+1}-(\threshgen+1)\utilizationA^{\threshgen+1}(1-\utilizationA)}{\rateSB\left(1-\utilizationA^{\threshgen+1}\right)+\rateA\utilizationA^{\threshgen}(1-\utilizationA)}\right) & \text{ if } \utilizationA\neq 1 \\ \frac{\threshgen}{2\rateSA} +\frac{1}{2}\frac{\threshgen+2}{\rateSB(\threshgen+1)+\rateA} & \text{ if } \utilizationA=1.
\end{dcases}
\end{equation}
For any candidate threshold $\threshgen<\naor$, the previous calculations give the expected payoff from remaining in the boundary state $(\threshgen,0)$. 
This payoff is nonnegative if and only if $\threshgen\leq\mAstar(\rewardB)$, where the generic $\mAstar(\valb)$ is defined as the unique nonnegative solution of \cref{eq:m-A}.
Because $\eqthreshB(0)$ is the largest integer number of $\typeA$-customers that Bernard is willing to tolerate, we obtain $\eqthreshB(0)=\floor{\mAstar(\rewardB)}$.
\end{proof}

\begin{figure}[!ht]
    \centering
\begin{tikzpicture}[
    x=1.42cm, y=1.42cm,
    >={Stealth[length=2mm]},
    font=\normalsize,
    transition/.style={-{Stealth[length=1.7mm]},line width=.52pt},
    gridline/.style={gray!45,thin,dashed},
    state/.style={circle,draw=black,fill=white,inner sep=0pt,minimum size=6pt,line width=.65pt}
]

\definecolor{continuationblue}{RGB}{167,223,249}
\definecolor{renegingred}{RGB}{255,192,196}
\def\yTop{5.80}

\fill[renegingred] (0,-.34) rectangle (7,\yTop);
\fill[continuationblue]
  (0,-.34)--(6,-.34)--(6,1.78)--(5,2.95)--(4,4.04)
  --(3,5.10)--(2,\yTop)--(0,\yTop)--cycle;

\foreach \x in {1,...,7} {
  \draw[gridline] (\x,-.34)--(\x,\yTop);
}
\foreach \y in {0,0.86,1.78,2.95,4.04,5.10} {
  \draw[gridline] (0,\y)--(7,\y);
}

\draw[blue!90!black,line width=1.1pt]
   (2,\yTop)--(3,5.10)--(4,4.04)--(5,2.95)--(6,1.78)--(6,-.34);

\foreach \yy/\maxx in {5.10/3,4.04/4,2.95/5,1.78/6,0.86/6,0/6} {
  \foreach \rightt in {1,...,\maxx} {
    \pgfmathtruncatemacro{\leftt}{\rightt-1}
    \draw[transition] (\leftt+.16,\yy+.065)--(\rightt-.16,\yy+.065);
    \draw[transition] (\rightt-.16,\yy-.065)--(\leftt+.16,\yy-.065);
    \ifnum\rightt=1
      \node[above,inner sep=0.6pt] at (\leftt+.5,\yy+.075) {$r$};
    \else
      \node[above,inner sep=0.6pt] at (\leftt+.5,\yy+.075) {$p$};
    \fi
    \node[below,inner sep=0.6pt] at (\leftt+.5,\yy-.075) {$q$};
  }
}
\foreach \yy/\xx/\zz in {5.10/3/4,4.04/4/5,2.95/5/6}{
    \draw[transition] (\xx+.16,\yy+.065)-- (\zz-.16,\yy+.065);
    \node[above,inner sep=0.6pt] at (\xx+.5,\yy+.075){$p$};
}

\foreach \upper/\lower in {5.10/4.04,4.04/2.95,2.95/1.78,1.78/0.86,0.86/0} {
  \draw[transition] (.11,\upper-.18)--(.11,\lower+.18);
  \pgfmathsetmacro{\mid}{(\upper+\lower)/2}
  \node[right,inner sep=1pt] at (.15,\mid) {$s$};
}

\foreach \yy in {0,0.86,1.78} {
  \draw[transition] (6.14,\yy+.015)
       .. controls (6.66,\yy-.04) and (6.78,\yy+.30) .. (6.48,\yy+.39)
       .. controls (6.24,\yy+.48) and (6.05,\yy+.34) .. (6.03,\yy+.19);
  \node[right,inner sep=0pt] at (6.49,\yy+.39) {$p$};
}

\foreach \yy in {0,0.86,1.78,2.95,4.04,5.10} {
  \foreach \xx in {0,...,7} {
    \node[state] at (\xx,\yy) {};
  }
}

\draw[-{Stealth[length=2.3mm]},line width=.8pt] (0,-.39)--(8.50,-.39);
\draw[-{Stealth[length=2.3mm]},line width=.8pt] (0,-.39)--(0,6.16);
\node[above left,inner sep=1pt] at (0,6.13) {$n_B$};
\node[right,inner sep=2pt] at (8.50,-.39) {$n_A$};

\foreach \x/\lab in {0/0,1/1,2/2,3/3,4/{\cdots},6/{N_A}} {
  \node[below,inner sep=3pt] at (\x,-.39) {$\lab$};
}
\foreach \y/\lab in {0/0,0.86/1,2.95/{b-2},4.04/{b-1},5.10/b} {
  \node[left,inner sep=5pt] at (0,\y) {$\lab$};
}
\node[left] at (-.32,2.36) {$\vdots$};

\draw[black!70,fill=white,line width=.45pt] (7.25,2.05) rectangle (9.85,6.05);
\node[anchor=west] at (7.56,5.68) {$p=\frac{\lambda_A}{\lambda_A+\mu_A},$};
\node[anchor=west] at (7.56,5.13) {$q=\frac{\mu_A}{\lambda_A+\mu_A},$};
\node[anchor=west] at (7.56,4.58) {$r=\frac{\lambda_A}{\lambda_A+\mu_B},$};
\node[anchor=west] at (7.56,4.03) {$s=\frac{\mu_B}{\lambda_A+\mu_B}.$};
\fill[continuationblue] (7.35,3.00) rectangle (7.74,3.38);
\node[anchor=west,font=\scriptsize] at (7.83,3.27) {continuation region};
\node[anchor=west,font=\scriptsize] at (7.83,3.06) {$n_A\leq K(n_B)$};
\fill[renegingred] (7.35,2.25) rectangle (7.74,2.63);
\node[anchor=west,font=\scriptsize] at (7.83,2.52) {reneging region};
\node[anchor=west,font=\scriptsize] at (7.83,2.31) {$n_A>K(n_B)$};

\end{tikzpicture}
    \caption{State transitions for a tagged class-$\typeB$ customer $\cust$.
    The blue region is the equilibrium continuation region
    $\{(\nAi,\nBi):\nAi\leq\eqthreshB(\nBi)\}$ and the red region is the reneging region. 
    For $\nAi>0$, an $\typeA$-arrival moves
    the state one step to the right with probability $p$, whereas an $\typeA$-service completion moves it one step to the left
    with probability $q$. 
    When $\nAi=0$, an $\typeA$-arrival
    moves the state to $(1,\nBi)$ with probability $r$, whereas a $\typeB$-service completion moves it to $(0,\nBi-1)$ with
    probability $s$. 
    At $\nAi=\naor$, $\typeA$-arrivals do not join; therefore they leave the state unchanged.}
    \label{fig:equilibrium-transitions}
\end{figure}
We now use the characterization of $\eqthreshB(0)$ to obtain 
recursively $\eqthreshB(1),\eqthreshB(2),\ldots$
\cref{fig:equilibrium-transitions} illustrates the resulting
state dynamics and the recursive relation between successive
$\typeB$-levels.

\begin{proposition}
\label{prop:RecursiveK}
Consider a symmetric Markov perfect equilibrium in stationary strategies, and let $\eqthreshB(\argdot)$ be its threshold function as characterized in \cref{thm:BThrsEq}.
Define $\valb_{-1}\coloneqq\rewardB$, and for every $\genb\geq0$ let $\valb_{\genb}\coloneqq\valueB(0,\genb)=\max\{0,\bellman(0,\genb)\}$. Then 
\begin{equation}
\label{eq:threshold(n-B)}
\eqthreshB(\genb)=
\begin{dcases}
-1 
& \text{if } \frac{\valb_{\genb-1}}{\CostB}<\frac{1}{\rateSB},\\
\floor*{\mAstar(\valb_{\genb-1})} 
& \text{if }
\frac{1}{\rateSB}\leq \frac{\valb_{\genb-1}}{\CostB}<\Etime^{\naor}(\naor,0),\\
\naor 
& \text{if }
\frac{\valb_{\genb-1}}{\CostB}\geq \Etime^{\naor}(\naor,0).
\end{dcases}
\end{equation}

Moreover, if $\threshgen=\eqthreshB(\genb)$, we have
\begin{equation}
\label{eq:RecValb}
\valb_{\genb}=
\begin{dcases}
0 
& \text{if }
\threshgen=-1,\\
\Prewardk(0,0)\valb_{\genb-1}-\CostB \Etimek(0,0) & \text{if }
0\leq \threshgen<\naor, \\
\valb_{\genb-1}-\frac{\CostB}{\rateSB}\sum_{j=0}^{\naor} \utilizationA^j & \text{if }
\threshgen=\naor.
\end{dcases}
\end{equation}
For $0\leq \threshgen <\naor$, the quantities $\Prewardk(0,0)$ and $\Etimek(0,0)$ are given by \eqref{eq:Preward1} and \eqref{eq:Etime1}, respectively. 
The case $\threshgen=\naor$ corresponds to a strategy without reneging, for which the service probability is $1$ and the expected service time is
\begin{equation}
\Etime^{\naor}(0,0) = 
\begin{dcases}
\frac{1-\utilizationA^{\naor+1}}{\rateSB(1-\utilizationA)} & \text{if } \utilizationA\neq 1\\
\frac{\naor+1}{\rateSB} &\text{if }  \utilizationA=1.
\end{dcases}    
\end{equation}
\end{proposition}

\begin{proof}
Let $\genb=\nBi$.
If $\nBi=0$, the result follows from \cref{lem:CharcK0}.

Now, let $\nBi\geq1$. By \cref{thm:BThrsEq}, as long as Bernard
remains in the system, no $\typeB$-customer ahead of him reneges.
Therefore, until Bernard reneges or a $\typeB$-service completion
occurs, the evolution is identical to the problem with no
$\typeB$-customer ahead of him considered in \cref{lem:CharcK0}.

At a $\typeB$-service completion, Bernard moves to state
$(0,\nBi-1)$ and receives the continuation value $\valb_{\nBi-1}$.
Hence, his stopping problem at level $\nB$ is identical to the problem at level $0$, with $\rewardB$ replaced by $\valb_{\nBi-1}$.
This gives the stated characterization of $\eqthreshB(\nB)$.

If $\threshgen=\eqthreshB(\nBi)$ and $0\leq \threshgen \naor$, starting from $(0,\nBi)$ Bernard reaches
$(0,\nBi-1)$ before reneging with probability $\Prewardk(0,0)$,
and the expected time until either event is $\Etimek(0,0)$.
Therefore,
\[
\valb_{\nB}
=
\Prewardk(0,0)\valb_{\nBi-1}
-\CostB\Etimek(0,0).
\]
If $\threshgen=\naor$, Bernard never reneges before the next $\typeB$-service completion. By \cref{le:E-time}, the expected time until this event is $\rateSB^{-1}\sum_{j=0}^{\naor} \utilizationA^j$, Hence,
\[\valb_{\nB} =
\valb_{\nBi-1} -\frac{\CostB}{\rateSB}\sum_{j=0}^{\naor} \utilizationA^j.\]
Finally, if $\eqthreshB(\nBi)=-1$, Bernard reneges immediately and
$\valb_{\nBi}=0$.
\end{proof}

The previous proposition provides a necessary recursive characterization of any symmetric stationary equilibrium:
\begin{equation}
\label{eq:v-1-arrow} 
\rewardB=\valb_{-1}\to \eqthreshB(0)\to \valb_{0}\to \eqthreshB(1)\to \valb_{1}\to\dots
\end{equation}
We now show that this recursion indeed generates such an equilibrium.

\begin{proposition}
\label{prop:MPEexistence}
Let $\valb_{-1}\coloneqq\rewardB$.
For $\genb\geq0$, recursively define $\eqthreshB(\genb)$ and $\valb_{\genb}$ using the formulas in \cref{prop:RecursiveK}.
Consider the stationary symmetric strategy profile $\strat$ in which $\typeA$-customers follow the threshold strategy in \cref{eq:Naor}, and each $\typeB$-customer $\cust$ in state $(\nAi,\nBi)$ joins (or remains in) the system if and only if $\nAi\leq\eqthreshB(\nBi)$.
Then $\strat$ is a Markov perfect equilibrium.
\end{proposition}

\begin{proof}
We prove by induction on $\nBi$ that,  in state $(0,\nBi)$, Bernard's optimal threshold is $\eqthreshB(\nBi)$ 
and $\valb_{\nBi}$ is the corresponding optimal value.

For $\nBi=0$, the same coupling argument used in the proof of \cref{thm:BThrsEq} shows that Bernard's optimal continuation payoff is weakly decreasing in the number of
$\typeA$-customers ahead. 
Hence, an optimal stopping strategy has a
threshold form, and the calculations in \cref{lem:CharcK0} show that this threshold is  $\eqthreshB(0)$. 
In particular,
\begin{equation*}
0\leq \valb_{0}\leq \valb_{-1}=\rewardB,   
\end{equation*}
where the first inequality is due to the fact that  Bernard can always stop immediately, whereas the second follows directly from the definition of
$\valb_{0}$.

Now fix $\nBi\geq1$ and suppose that, for every $\initnum<\nBi$, $\eqthreshB(\initnum)$ is Bernard's optimal threshold and $\valb_{\initnum}$ is the optimal continuation value in state $(0,\initnum)$.
We have
\[
0\leq\valb_{\initnum}\leq\valb_{\initnum-1},
\quad \initnum=0,\ldots,\nBi-1.
\]
For a fixed terminal reward $\valb$, consider the problem of \cref{lem:CharcK0} ($\nBi=0$). 
For any candidate threshold $\threshgen \geq 0$, the payoff from remaining in its boundary state $(\threshgen,0)$ is $
\Prewardk(\threshgen,0)\valb-\CostB\Etimek(\threshgen,0)$,
which is increasing in $\valb$. 
Hence, the optimal threshold is weakly increasing in the terminal reward. 
Because $\valb_{\nBi-1} \leq \valb_{\nBi-2}$, we obtain
\[
\eqthreshB(\nBi)\leq\eqthreshB(\nBi-1).
\]
Therefore, by induction,
$\eqthreshB(\nBi)\leq\eqthreshB(\initnum)$ for every $\initnum<\nBi$.

Consequently, whenever  Bernard's strategy prescribes remaining in the system at level $\nBi$, every $\typeB$-customer ahead of him also remains in the system. 
Hence, until Bernard reneges or a $\typeB$-service completion occurs, his problem is identical to the  problem considered above for $\nBi=0$.
The difference is that a $\typeB$-service completion now moves Bernard to state $(0,\nBi-1)$, whose continuation value is
$\valb_{\nBi-1}$. Therefore, remaining is optimal if and only if $\nAi\leq \eqthreshB(\nBi)$.

Consider now a state $(\nAi,\nBi)$ with $\nAi>\eqthreshB(\nBi)$.
Suppose that Bernard deviates and remains in the system. 
If no $\typeB$-customer ahead of him reneges, Bernard faces the same  problem as above, in which reneging is optimal because $\nAi>\eqthreshB(\nBi)$.
Suppose instead that some $\typeB$-customers ahead of Bernard renege. 
Let $\initnum<\nBi$ be the number of class-$\typeB$ customers left ahead of Bernard after all instantaneous reneging decisions have been resolved. 
The number of $\typeB$-customers ahead of the last customer who reneged (Belinda) is exactly $\initnum$. 
Belinda and Bernard have the same number $\nAi$ of $\typeA$-customers ahead of them.
Since Belinda  follows the prescribed strategy, her decision to renege implies $\nAi>\eqthreshB(\initnum)$. 
Therefore, Bernard is  left in state $(\nAi,\initnum)$, and the induction hypothesis implies that reneging is optimal in this state. 
Hence, Bernard cannot profit from deviating, so reneging is optimal if and only if $\nAi>\eqthreshB(\nBi)$.

Therefore, the recursively defined threshold $\eqthreshB(\nBi)$ is optimal at level $\nBi$.

Let $\threshgen=\eqthreshB(\nBi)$. 
If $\threshgen\geq0$, the corresponding optimal value in
state $(0,\nBi)$ is
\[
\valb_{\nBi}
=
\Prewardk(0,0)\valb_{\nBi-1}
-
\CostB\Etimek(0,0),
\]
whereas if $\threshgen=-1$, Bernard reneges  immediately and $\valb_{\nBi}=0$.
Because  $\Prewardk(0,0)\leq1$ and $\Etimek(0,0)\geq0$, we obtain $0\leq\valb_{\nBi}\leq\valb_{\nBi-1}$. Thus, the induction is complete.

This shows that the recursively defined strategy is optimal for every $\typeB$-customer in every state. 
Because $\typeA$-customers follow their optimal Naor threshold strategy independently of the behavior of $\typeB$-customers, the resulting stationary symmetric strategy profile is a Markov perfect equilibrium.
\end{proof}

\begin{remark}
The threshold function $\eqthreshB$ can equivalently be represented by the threshold in the number of $\typeB$-customers ahead. 
Consider Bernard's  continuation payoff $\bellman(\argdot,\argdot)$ introduced in the proof of \cref{thm:BThrsEq}. 
For each $\gena\in\{0,1,\dots,\naor\}$, define
\begin{equation}
\label{eq:J-star}
\threshB^*(\gena)\coloneqq\min\{\genb\geq0 \colon\bellman(\gena,\genb)<0\}
\end{equation}
Because  $\bellman(\nAi,\nBi)$ is weakly decreasing in $\nBi$, Bernard remains in the system if and only if 
$\nBi< \threshB^*(\nAi)$.
Moreover, because $\bellman(\nAi,\nBi)$ is weakly decreasing in $\nAi$, the function $\threshB^*$ is weakly decreasing. 
Hence, the two representations are equivalent, that is:
\begin{equation}
\label{eq:equivalent}   \nAi\leq \eqthreshB(\nBi) \iff \nBi<\threshB^*(\nAi). 
\end{equation}
\end{remark}

The following corollary deals with the case $\rewardB/\CostB\geq \Etime^{\naor}(\naor,0)$.
If this condition is satisfied, when there are no $\typeB$-customers ahead of Bernard, he prefers to remain even in the most congested $\typeA$-state $\naor$. 
Equivalently, by \cref{lem:CharcK0}, $\eqthreshB(0)=\naor$. The next corollary determines the largest $\genb$ for which $\eqthreshB(\genb)=\naor$.

\begin{corollary}
\label{cor:KNA}
Suppose that $\rewardB/\CostB \geq \Etime^{\naor}(\naor,0)$. 
Define
\begin{equation}
\label{eq:Vb}
\Veq \coloneqq 
\begin{dcases}
\frac{\rateSB}{\rateSA\parens*{1 - \utilizationA^{\naor+1}}} \parens*{\frac{\rewardB}{\CostB}(\rateSA-\rateA)-\naor+\frac{\utilizationA\parens*{1-\utilizationA^{\naor}}}{1-\utilizationA}}-1 & \text{ if } \utilizationA\neq 1, \\ \frac{\rewardB}{\CostB}\frac{\rateSB}{\naor+1} -\frac{\naor\rateSB}{2\rateSA}-1
& \text{ if } \utilizationA=1.
\end{dcases}
\end{equation}
Then
\begin{equation}
\label{eq:equivalent-2} \eqthreshB(\genb)=\naor \iff \genb\leq\Veq.   
\end{equation}
Equivalently $\threshB^*(\naor)=\floor{\Veq}+1$.
\end{corollary}

\begin{proof}
Assume $\rewardB/\CostB\geq\Etime^{\naor}(\naor, 0)$. 
By \cref{lem:CharcK0}, $\eqthreshB(0)=\naor$. 
We now determine the values of $\genb$ for which the equality $\eqthreshB(\genb)=\naor$  holds. 
Let Bernard's state $(\nAi, \nBi)$ be equal to $(\naor,\genb)$. 
Because the number of $\typeA$-customers cannot exceed $\naor$, and any future $\typeB$-customer joins behind Bernard, his position cannot deteriorate.
Therefore, if it is optimal for him to remain in the system in this state, he will stay in the system until service completion. 
Hence,
\begin{equation}
\label{eq:Bjoinqueue}
\eqthreshB(\genb)=\naor \iff \rewardB-\CostB\EtimeNR(\naor,\genb)\geq 0,
\end{equation}
where $\EtimeNR(\naor,\genb)$ denotes the expected time Bernard spends in the system in state $(\naor,\genb)$, if the $\typeA$-customers use the Naor threshold strategy and the $\typeB$-customers ahead of Bernard do not renege.
By \cref{le:E-time}, this condition is equivalent to
\begin{equation}
\label{eq:Bjoinqueue2}       \frac{\rewardB}{\CostB}\geq
\begin{cases} \displaystyle{\frac{1}{\rateSA-\rateA} \bracks*{(\genb+1)\frac{\rateSA}{\rateSB}\parens*{1 - \utilizationA^{\naor+1}}+  \naor-\frac{\utilizationA\parens*{1-\utilizationA^{\naor}}}{1-\utilizationA} }} 
& \text{ if } \utilizationA\neq 1, \\
\\
\displaystyle{\frac{\naor(\naor+1)}{2\rateSA}+(\genb+1)\frac{1}{\rateSB}(1+\naor)} & \text{ if } \utilizationA=1.
\end{cases}
\end{equation}
Hence, if $\utilizationA\neq 1$, then
\begin{align}
&(\genb+1)\frac{\rateSA}{\rateSB}\parens*{1 - \utilizationA^{\naor+1}}
\le \frac{\rewardB}{\CostB}(\rateSA-\rateA)-\naor+\frac{\utilizationA\parens*{1-\utilizationA^{\naor}}}{1-\utilizationA},
\intertext{that is,}
&\genb+1
\le \frac{\rateSB}{\rateSA\parens*{1 - \utilizationA^{\naor+1}}} \parens*{\frac{\rewardB}{\CostB}(\rateSA-\rateA)-\naor+\frac{\utilizationA\parens*{1-\utilizationA^{\naor}}}{1-\utilizationA}}.    
\end{align}
If $\utilizationA=1$, then
\begin{align}
&(\genb+1)\frac{1}{\rateSB}(1+\naor)\le\frac{\rewardB}{\CostB} -\frac{\naor(\naor+1)}{2\rateSA},
\intertext{that is,}
&\genb+1\le\frac{\rewardB}{\CostB}\frac{\rateSB}{\naor+1} -\frac{\naor\rateSB}{2\rateSA}.   \end{align}
Then, defining $\Veq$ as in \eqref{eq:Vb}, we obtain 
\begin{equation*}
\eqthreshB(\genb)=\naor \iff \genb \le \Veq.
\qedhere
\end{equation*}
\end{proof}

To illustrate the variety of equilibrium shapes allowed by \cref{thm:BThrsEq}, we conclude this section with a few numerical examples. 
Recall that Bernard remains in state $(\nAi,\nBi)$ if and only if $\nAi\leq\eqthreshB(\nBi)$. 
Hence, the graph of $\eqthreshB$ represents the boundary of the equilibrium continuation region.

\cref{fig:equilibrium-threshold-shapes} shows four representative examples. 
In all panels, the threshold is initially constant at $\naor$, so that Bernard is willing to tolerate the maximal number of $\typeA$-customers when $\nBi$ is sufficiently small. 
Beyond this initial flat region, however, the equilibrium threshold may exhibit substantially different behavior.

Panels~(a)--(c) illustrate three different patterns arising under heterogeneous service rates. 
Thus, apart from the monotonicity established in \cref{thm:BThrsEq}, the shape of the equilibrium boundary can depend substantially on the model parameters.
Panel~(d) illustrates the special case of equal service rates. 
As we will show in the next section, once the threshold leaves the truncation at $\naor$, it decreases linearly, reflecting the fact that Bernard's equilibrium decision depends only on the total number of customers ahead of him.

\begin{figure}[!ht]
    \centering
    \includegraphics[width=\textwidth]{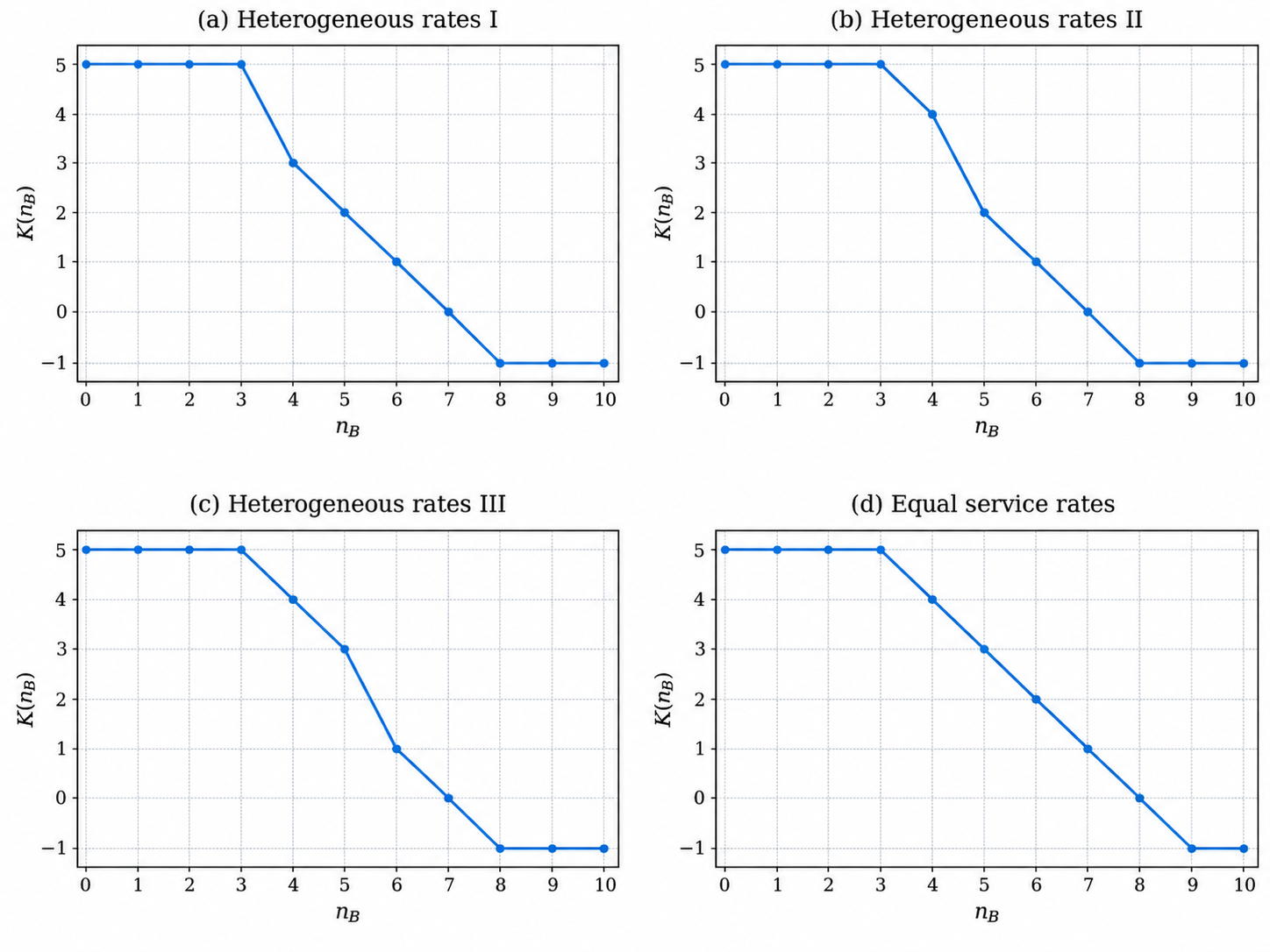}
    \caption{Examples of equilibrium-threshold functions
$\eqthreshB$. 
In all panels, the threshold is initially constant at
$\naor$. 
Panels~(a)--(c) illustrate different equilibrium shapes under
heterogeneous service rates, while panel~(d) illustrates the
equal-service-rate case, in which the threshold is linear once it
leaves the truncation at $\naor$.
In all four panels, we set $\CostA=\CostB=1$, $\rewardA=7$, and $\rateSA=0.8$, so that $\naor=5$. 
Panel~(a) uses $\rateA=0.1$, $\rateSB=0.6$, and $\rewardB=15$; 
panel~(b) uses $\rateA=0.3$, $\rateSB=0.6$, and $\rewardB=20.5$; 
panel~(c) uses $\rateA=0.15$, $\rateSB=0.6$, and $\rewardB=17$; 
finally, panel~(d) uses $\rateA=0.1$, $\rateSB=\rateSA=0.8$, and $\rewardB=13$.}
    \label{fig:equilibrium-threshold-shapes}
\end{figure}

%
%

\section{Equal Service Rates}
\label{se:equal-rates}

We now consider the special case in which the two classes have the same
service rate, $\rateSA=\rateSB=\rateS$. This is the setting studied by
\citet{AfeSar:MSfrth}. 
Their state representation differs from ours:
in their framing a low-priority customer is described by her position among low-priority customers and by her position among all customers. 
The two representations are, however, equivalent.
Indeed, if Bernard has $\nAi$ $\typeA$-customers and $\nBi$
$\typeB$-customers ahead of him, then his low-priority position and
system position in the notation of \citet{AfeSar:MSfrth} are, respectively, $m=\nBi+1$
and $n=\nAi+\nBi+1$.

We show that the recursive characterization of \cref{prop:RecursiveK} simplifies considerably when the service rates are equal. In particular, the equilibrium depends only on the total number of customers ahead, recovering the equilibrium structure characterized by \citet{AfeSar:MSfrth}.
Define
\begin{equation}
\label{eq:brakets}
\bracks*{x}_{h}^{k} \coloneqq \max\braces*{h,\min\braces*{k,x}}.
\end{equation}

\begin{theorem}
\label{thm:EqualRates}
Suppose that $\rateSA=\rateSB=\rateS$. 
Then there exists an integer 
$\nbarB\geq-1$ such that 
\begin{equation}
\label{eq:threshold-linear}  \eqthreshB(\genb) = 
\bracks*{\nbarB-\genb}_{-1}^{\naor}, 
\end{equation}
where the function $\eqthreshB$ is defined as in \cref{thm:BThrsEq}.
\end{theorem}

In other words, \cref{thm:EqualRates} states that Bernard stays in the system if 
$\nAi+\nBi \le \nbarB$.

\begin{proof}
We derive the result directly from the recursive characterization of \cref{prop:RecursiveK}. 
For every $\threshgen \ge 0 $ define
\begin{equation}
\label{eq:s-d}
\sumrhok\coloneqq\sum_{j=0}^{\threshgen}\utilizationA^j
\quad\text{and}\quad
\sumsk\coloneqq\frac{\CostB}{\rateS} \sum_{j=0}^{\threshgen}\sumrho_{j},    
\end{equation}
with $\sums_{-1}\coloneqq0$. 
After some algebra we obtain 
\begin{align}
\label{eq:skgeneq}
\sumrhok 
&=  
\begin{dcases}
    \frac{1-\utilizationA^{\threshgen+1}}{1-\utilizationA} &\text{ if } \utilizationA\neq 1, \\ \threshgen+1 &\text{if } \utilizationA=1,
\end{dcases} \\
\label{eq:dkgeneq}
\sumsk
&=
\begin{dcases}
\frac{\CostB}{\rateS(1-\utilizationA)}
\parens*{\threshgen+1 -\utilizationA\frac{1-\utilizationA^{\threshgen+1}}{1-\utilizationA}} &\text{if } \utilizationA\neq 1, \\ 
\frac{\CostB(\threshgen+1)(\threshgen+2)}{2\rateS}
&\text{if }\utilizationA = 1.
\end{dcases}
\end{align}
This is equivalent to
\begin{equation}
\label{eq:dk-mu-CB}   \frac{\sumsk\rateS}{\CostB} = 
\begin{dcases}
\frac{(\threshgen+1)-(\threshgen+2)\utilizationA+\utilizationA^{\threshgen+2}}{(1-\utilizationA)^2} & \text{if } \utilizationA\neq 1,\\
\frac{(\threshgen+1)(\threshgen+2)}{2} & \text{if } \utilizationA=1.
\end{dcases}
\end{equation}

By \cref{le:E-time}, it is simple to see that $\sums_{\naor}=\CostB \Etime^{\naor}(\naor,0)$ and $\sums_{0}=\CostB/\rateS$.

Considering \cref{eq:m-A} when $\rateSB=\rateSA=\rateS$, we see that $\mAstar(\valb)$ is the unique nonnegative solution $\nAstar$ of the following equation:
\begin{equation}
\label{eq:mstarsol}
\frac{\valb\rateS}{\CostB}=
\begin{dcases}
\frac{(\nAstar+1)-(\nAstar+2)\utilizationA+\utilizationA^{\nAstar+2}}{(1-\utilizationA)^2} & \text{if } \utilizationA\neq 1,\\
\frac{(\nAstar+1)(\nAstar+2)}{2} & \text{if } \utilizationA=1.
\end{dcases}    
\end{equation}
The \rhs of \cref{eq:mstarsol} is monotone in $\nAstar$; therefore $\floor{\mAstar(\valb)}=\threshgen$ if and only if
\begin{equation}
\begin{aligned}
\label{eq:diseqv}  \frac{(\threshgen+1)-(\threshgen+2)\utilizationA+\utilizationA^{\threshgen+2}}{(1-\utilizationA)^2}
&\le
\frac{\valb\rateS}{\CostB}
<
\frac{(\threshgen+2)-(\threshgen+3)\utilizationA+\utilizationA^{\threshgen+3}}{(1-\utilizationA)^2} 
&\text{if } \utilizationA\neq 1,\\
\frac{(\threshgen+1)(\threshgen+2)}{2}
&\le
\frac{\valb\rateS}{\CostB}
<
\frac{(\threshgen+2)(\threshgen+3)}{2} 
&\text{if } \utilizationA=1.
\end{aligned}    
\end{equation}
Combining  \cref{eq:diseqv,eq:dk-mu-CB}, we have that $\floor{\mAstar(\valb)}=\threshgen$ if and only if
\begin{equation}
\sumsk\leq \valb<\sums_{\threshgen+1}.    
\end{equation}
We can now rewrite the threshold characterization in \cref{prop:RecursiveK} as follows:
\begin{equation}
\label{eq:Kbsamerate}
\eqthreshB(\genb)=
\begin{dcases}
-1, & \text{if }\valb_{\genb-1}<\sums_{0},\\
\threshgen, & \text{if }\sumsk\le \valb_{\genb-1}<\sums_{\threshgen+1},
       \quad 0\le \threshgen<\naor,\\
\naor, & \text{if } \valb_{\genb-1}\ge \sums_{\naor}.
\end{dcases}    
\end{equation}

Assume that $\eqthreshB(\genb)=\threshgen$, with $0\le \threshgen< \naor$. 
Under equal service rates, \cref{eq:Preward1,eq:Etime1} become
\begin{align}
\label{eq:Pk-equal rates}    
\Prewardk(0,0)
&=
\begin{dcases}\frac{1-\utilizationA^{\threshgen+1}}{1-\utilizationA^{\threshgen+2}} & \text{if } \utilizationA\neq 1,\\
\frac{\threshgen+1}{\threshgen+2} & \text{if } \utilizationA= 1,
\end{dcases}
\\
\label{eq:Ek-equal rates} 
\Etimek(0,0)
&=
\begin{dcases}\frac{1-\utilizationA^{\threshgen+1}-(\threshgen+1)\utilizationA^{\threshgen+1}(1-\utilizationA^{\threshgen})}{\rateS(1-\utilizationA)(1-\utilizationA^{\threshgen+2})} & \text{if } \utilizationA\neq 1,\\
\frac{\threshgen+1}{2\rateS} & \text{if } \utilizationA= 1,
\end{dcases}
\end{align}
which imply 
\begin{equation}
\label{eq:Pk-Ek-s-d}
\Prewardk(0,0)=\frac{\sumrhok}{\sumrho_{\threshgen+1}},
\quad
\CostB\Etimek(0,0)=\frac{\sumrhok}{\sumrho_{\threshgen+1}}\sumsk-\sums_{\threshgen-1}.
\end{equation}
Substituting these expressions into \cref{eq:RecValb}, we get
\begin{equation}
\label{eq:v-b-s-d}  
\valb_{\genb}=\sums_{\threshgen-1}+ \frac{\sumrhok}{\sumrho_{\threshgen+1}}(\valb_{\genb-1}-\sumsk).
\end{equation}
From the assumption $\eqthreshB(\genb)=\threshgen$ and \cref{eq:Kbsamerate}, we have $\sumsk\le\valb_{\genb-1}<\sums_{\threshgen+1}$. 
Hence, applying these inequalities to the previous equation for $\valb_{\genb}$ we obtain
\[
\sums_{\threshgen-1}\leq \valb_\genb< \sums_{\threshgen-1}+ \frac{\sumrhok}{\sumrho_{\threshgen+1}}(\sums_{\threshgen+1}-\sumsk)=\sumsk,
\]
where the last equality follows from $\sums_{\threshgen+1}-\sumsk=\CostB \sumrho_{\threshgen+1}/\rateS$. 
Then, by the threshold characterization in \cref{eq:Kbsamerate}, $\eqthreshB(\genb+1)=\threshgen-1$.

Suppose now that $\eqthreshB(\genb)=\naor$, so $\valb_{\genb-1}\ge \sums_\naor $. 
The recursion in \cref{eq:RecValb} gives 
\[
\valb_\genb=\valb_{\genb-1}-\frac{\CostB}{\rateS}\sumrho_{\naor}=\valb_{\genb-1}-(\sums_{\naor}-\sums_{\naor-1})
\]
and $\valb_\genb\ge \sums_{\naor-1} $. Therefore, $\eqthreshB(\genb+1)$ is either $\naor$ or $\naor-1$.

Finally, if $\eqthreshB(\genb)=-1$, then $\valb_\genb=0<\sums_{0}$ and, consequently, $\eqthreshB(\genb+1)=-1$.

This shows that the threshold $\eqthreshB(\genb)$ takes value $\naor$ for small values of $\genb$, then, at some point, it starts decreasing with slope $-1$, until it reaches value $-1$.

Because  $\valb_\genb$ is strictly decreasing, the threshold strategy $\eqthreshB$ can not remain constantly equal to $\naor$.

Finally, If $\eqthreshB(0)<\naor$, the result follows by setting $\nbarB\coloneqq\eqthreshB(0)$.
Alternatively, let $\floor{\Veq}$ be the largest
integer such that $\eqthreshB(\floor{\Veq})=\naor$ (which exists because of  \cref{cor:KNA}), and set $\nbarB\coloneqq\naor+\floor{\Veq}$. 
In either case,
\[
\eqthreshB(\genb)=\max\{-1,\min\{\naor,\nbarB-\genb\}\}
=[\nbarB-\genb]_{-1}^{\naor}.
\qedhere
\]
\end{proof}

\begin{corollary}
\label{cor:EqualRatesThreshold}
Suppose that $\rateSA=\rateSB=\rateS$, and let $\nbarB$ be the threshold
defined in \cref{thm:EqualRates}. Then
\begin{equation}
\label{eq:NB-equal-rates}
\nbarB=
\begin{dcases}
-1 &\text{if } \frac{\rewardB}{\CostB}<\frac{1}{\rateS},\\
\floor{\mAstar(\rewardB)}
&\text{if } 
\frac{1}{\rateS}\leq \frac{\rewardB}{\CostB}
<
\Etime^{\naor}(\naor,0),
\\
\naor+\floor{\Veq}
&
\text{if } \frac{\rewardB}{\CostB}
\geq
\Etime^{\naor}(\naor,0),
\end{dcases}
\end{equation}
where 
$\Veq$ is defined as in \cref{eq:Vb}.
\end{corollary}

\begin{remark}
\label{rem:AfecheSarhangian}
The above characterization  recovers the equilibrium structure of \citet{AfeSar:MSfrth}. Recall that, under the correspondence
\[
m=\nBi+1,
\qquad
n=\nAi+\nBi+1,
\]
their high-priority threshold $\bar n_h$ corresponds to our $\naor$.

Therefore,
\[
\nAi+\nBi\leq\nbarB \quad\Longleftrightarrow\quad n\leq \nbarB+1.
\]
Hence, denoting by $L^{\mathrm{AS}}$ the maximum system position in which
a low-priority customer joins or remains in the system in
\citet{AfeSar:MSfrth}, we obtain
\[
L^{\mathrm{AS}}=\nbarB+1.
\]

Their second threshold, denoted here by $K^{\mathrm{AS}}$, is the maximum
low-priority position from which a customer never abandons. In our
notation, this corresponds to the largest $\nB+1$ for which
$\eqthreshB(\nB)=\naor$. Therefore,
\[
K^{\mathrm{AS}} =
\begin{cases}
0, & \nbarB<\naor, \\[2mm]
\nbarB-\naor+1, & \nbarB\geq\naor.
\end{cases}
\]
Using \cref{eq:NB-equal-rates}, this can equivalently be written as
\[
K^{\mathrm{AS}}
=
\begin{cases}
0,
&
\displaystyle
\frac{\rewardB}{\CostB}
<
\Etime^{\naor}(\naor,0),
\\[3mm]
\floor{\Veq}+1,
&
\displaystyle
\frac{\rewardB}{\CostB}
\geq
\Etime^{\naor}(\naor,0).
\end{cases}
\]

In particular, whenever $K^{\mathrm{AS}}>0$,
\[
L^{\mathrm{AS}}
=
\naor+K^{\mathrm{AS}},
\]
which is precisely the relation between the two equilibrium thresholds
obtained by \citet{AfeSar:MSfrth}.
Thus, under equal service rates, the two-dimensional threshold
characterization of \cref{thm:BThrsEq} collapses to a single threshold
on the customer's system position. Under the change of variables above,
this is exactly the equilibrium structure characterized by
\citet{AfeSar:MSfrth}.
\end{remark}

%
%

\section{Stability and Linear Bounds for the Equilibrium Strategy}
\label{se:unequal-service-rates}

\cref{se:equal-rates} shows that, when the two classes have identical service rates, the equilibrium strategy $\eqthreshB(\argdot)$ is linear, up to truncation.
Here we prove that this equilibrium strategy is locally stable around identical service rates, provided the equal-rate equilibrium has no indifference states.
Although a linear approximation is in general not  accurate for larger differences in service rates, we can build two linear thresholds that bound the equilibrium strategy from below and from above.
The distance between these bounds provides a sufficient condition on the parameters of the system for the  existence of a good linear approximation of the equilibrium-threshold function.

%
%

\subsection{Local Stability}
Fix $\rateSA=\rateS$ and let  $\rateSB=(1+\varepsilon)\rateS$ with $\varepsilon>-1$.
Let $\eqthreshBeps$ denote the equilibrium function characterized in \cref{se:equilibrium}.
For $\varepsilon=0$, \cref{thm:EqualRates} gives 
$\eqthreshB_{0}(\genb)=[\nbarB-\genb]_{-1}^{\naor}$.

\begin{theorem}
\label{thm:local-stability}
Suppose that the equilibrium under identical service rates is strict, \ie in any possible state, Bernard is never  indifferent between remaining in the system and reneging.
Then there exists $\varepsilon_{0}>0$ such that, for every
$\abs*{\varepsilon}<\varepsilon_{0}$,
\begin{equation}
\eqthreshBeps(\genb) = \eqthreshB_{0}(\genb),
\quad\text{for all }\genb \ge 0.
\end{equation}
\end{theorem}

\begin{proof}
For each fixed state, the continuation payoff is continuous
in $\varepsilon$. Indeed, by \cref{prop:RecursiveK}, the optimal
values are computed recursively, and at each level they
are obtained by maximizing---over finitely many thresholds---policies whose expected payoffs depend continuously on
the service rates.

By the strictness assumption, the signs of the continuation
payoffs at the equilibrium boundaries are preserved under
sufficiently small perturbations. 
Because \cref{thm:EqualRates} gives $\eqthreshB_{0}(\genb)=-1$ for all $\genb\ge \nbarB+1$, it is enough to consider the finitely many levels $\genb\in\braces*{0,\ldots,\nbarB+1}$.

Choosing $\varepsilon_{0}>0$ sufficiently small preserves
the equilibrium threshold at all these levels simultaneously.
In particular, $\eqthreshBeps(\nbarB+1)=-1$, and the
monotonicity of $\eqthreshBeps$ established in \cref{thm:BThrsEq}
implies $\eqthreshB(\genb)=-1$ for every $\genb\ge \nbarB+1$.

Hence $\eqthreshBeps(\genb)=\eqthreshB_{0}(\genb)$ for all $\genb\ge0$ and all $\abs*{\varepsilon}<\varepsilon_{0}$, completing the proof.
\end{proof}

Under a non-indifference condition, the theorem establishes local constancy of the equilibrium strategy, rather than merely convergence of an approximation error to zero: the equilibrium threshold remains unchanged under sufficiently small perturbations of the service rate, even though the underlying continuation values may change.

%
%

\subsection{Limits of Linear Approximation}

The local stability result does not imply that the equilibrium-threshold function  can always be accurately approximated by a linear function. 
We now show that, even if the slope and the intercept of a linear approximation are chosen optimally, its error can be made arbitrarily large.

For a given equilibrium-threshold function $\eqthreshB$, define
\begin{equation}
\label{eq:best-linear-error}
\error(\eqthreshB)\coloneqq
\inf_{\intercept,\slope\in\mathbb R} \sup_{\genb\ge0}
\abs*{\eqthreshB(\genb)-[\intercept+\slope \genb]_{-1}^{\naor}}.
\end{equation}
We allow the approximating function to take noninteger values, making this criterion more
permissive than an approximation by integer-valued thresholds.

\begin{proposition}
\label{prop:failure-linear}
For every $\bound>0$, there exist parameters of the
priority queueing model such that $\error(\eqthreshB)>\bound$.
\end{proposition}

\cref{prop:failure-linear} states that the equilibrium-threshold function cannot be uniformly approximated by a linear function with a parameter-independent error bound.

\begin{proof}
For each integer $\run\ge2$, consider 
\begin{equation}
\label{eq:parameters-N}
\rateA=\rateSA=\CostB=1, \quad  \rateSB=\frac{1}{\run}, \quad
\naor=\run, \quad \rewardB=\frac{3}{2}\run(\run+1).    
\end{equation}
Note that we can always choose $\typeA$-class parameters such that $\naor=\run$. 
Let $\eqthreshB_{\run}$ denote the equilibrium threshold for the above parameters.

Recall that, for the above parameters, \cref{eq:m-A} implies that $\mAstar(\valb)$ is the unique positive solution $\nAstar$ of the equation
\begin{equation}
\label{eq:vcharex}
\frac{\valb}{\run}= \frac{1}{2}(\nAstar+1)\left(\frac{\nAstar}{\run}+2\right)=(\nAstar+1)\left(1+\frac{\nAstar}{2\run}\right).
\end{equation}

By \cref{le:E-time},
\begin{equation}
\label{eq:E-N-A} 
\Etime^{\naor}(\naor,0)=\frac{\naor(\naor+1)}{2}+\naor(\naor+1)=\frac{3}{2}\run(\run+1)=\rewardB.
\end{equation}
So, applying  \cref{lem:CharcK0}, we have $\eqthreshB_{\run}(0)=\run$. 
Hence, by the recursion \cref{eq:RecValb} in \cref{prop:RecursiveK}, 
\begin{equation*}
\valb_{0}=\valb_{-1}-\frac{\CostB}{\rateSB}\sum_{j=0}^{\naor} \utilizationA^j=\frac{3}{2}\run(\run+1)-\run(\run+1)=\frac{\run(\run+1)}{2}.    
\end{equation*}
Then,
$$\frac{1}{\rateSB}=\run\leq \frac{\valb_{0}}{\CostB}=\frac{\run(\run+1)}{2}<\Etime^{\naor}(\naor,0)=\frac{3}{2}\run(\run+1)$$
and, by \cref{prop:RecursiveK}, $\eqthreshB(1)=\floor{\mAstar(\valb_{0})}$, where by \cref{eq:vcharex},
 $\mAstar(\valb_{0})$ is the unique positive solution $\nAstar$ of
\[
(\nAstar+1)\left(1+\frac{\nAstar}{2\run}\right) = \frac{\run+1 }{2}.
\]
Solving the above equation we have $\mAstar(\valb_{0})=(\sqrt{8\run^2+1}-(2\run +1))/2$,  which implies
\begin{equation}
\label{eq:first-threshold-asymptotic}
\frac{\eqthreshB_{\run}(1)}\run \xrightarrow[\run\to\infty]{}
\constone \coloneqq\sqrt2-1.
\end{equation}

Let $\threshgen_{1}\coloneqq\eqthreshB_{\run}(1)$. 
Specializing \cref{eq:Preward1} and \cref{eq:Etime1}, we obtain
\begin{equation}
\label{eq:P-k-1-E-k-1} 
\Preward^{\threshgen_{1}}(0,0)=\frac{\threshgen_{1}+1}{\run+\threshgen_{1}+1},
\quad
\Etime^{\threshgen_{1}}(0,0)=
\frac{\run(\threshgen_{1}+1)(\threshgen_{1}+2)}
{2(\run+\threshgen_{1}+1)}.
\end{equation}
Therefore, \cref{eq:RecValb} in \cref{prop:RecursiveK} gives
\begin{equation}
\label{eq:v-1=}    
\valb_{1}=
\frac{\run(\threshgen_{1}+1)(\run-\threshgen_{1}-1)}
{2(\run+\threshgen_{1}+1)}.
\end{equation}
Solving \cref{eq:vcharex} for $\valb=\valb_{1}$, we obtain 
\begin{equation}
\label{eq:m-*-v-1}
\mAstar(\valb_{1})=\frac{\sqrt{(2\run-1)^2+8\valb_{1}}-(2\run+1)}{2}.
\end{equation}
Then, 
\begin{equation}
\label{eq:second-threshold-asymptotic}
\lim_{\run\to \infty}\frac{\eqthreshB_{\run}(2)}\run= 
\lim_{\run\to\infty} \frac{ \sqrt{(2 - 1/\run)^2+8\valb_{1}/\run^2} -(2+1/\run)}{2}=\sqrt{1+\constone^2}-1\eqqcolon \consttwo.
\end{equation}
where in the last equality we use $\lim_{\run\to \infty} \valb_{1}/\run^2=\constone^2/2$.

Suppose, by contradiction, that
$\error(\eqthreshB_{\run}) = \smalloh(\run)$ along a subsequence.
For each $\run$, choose coefficients $\intercept_{\run},\slope_{\run}\in\mathbb R$
such that the truncated affine function
\begin{equation}
\label{eq:truncated affine}
\truncaff_{\run}(\genb)\coloneqq[\intercept_{\run}+\slope_{\run}\genb]_{-1}^{\run}
\end{equation}
has an uniform approximation error
\begin{equation}
\label{eq:uniform-error}
\approxerror_{\run}\coloneqq
\sup_{\genb\ge0} \abs*{\eqthreshB_{\run}(\genb)-\truncaff_{\run}(\genb)} 
= \smalloh(\run).
\end{equation}
It follows that
\begin{equation}
\label{eq:limit-1-2-3}   
\lim_{\run\to\infty}\frac{\truncaff_{\run}(\genb)}{\run} = \const_{\genb}, \quad\text{for }\genb=0,1,2,
\end{equation}
where $\const_{0}\coloneqq1$.

Because $0<\consttwo<\constone<1$,  the truncation is inactive at $\genb=1,2$ for sufficiently large $\run$.
Consequently,
\begin{equation}
\label{eq:alpha-N=}  
\intercept_{\run} = 2 \truncaff_{\run}(1) - \truncaff_{\run}(2). 
\end{equation}
Hence,
\begin{equation}
\label{eq:alpha-N/N}  
\frac{\alpha_{\run}}{\run}
\to 2\constone-\consttwo < 1.
\end{equation}
We know that $2\constone-\consttwo > 0$, so the truncation is inactive also at $\genb=0$, which implies that $\truncaff_{\run}(0) = \intercept_{\run}$. 
Therefore,
\begin{equation}
\label{eq:G-N/N=}
\frac{\truncaff_{\run}(0)}{\run} \to 2\constone-\consttwo < 1.
\end{equation}
contradicting $\truncaff_{\run}(0)/\run \to \const_{0}=1$.
Thus, $\liminf_{\run\to\infty}\error(\eqthreshB_{\run})/\run > 0$, which proves the proposition.
\end{proof}

\cref{prop:failure-linear} shows that the failure of linear approximation is not limited to a particular choice of slope. 
Even when the slope and intercept are optimized, the approximation error may be arbitrarily large.
We next identify a sufficient condition under which a linear approximation is guaranteed to be accurate.

%
%

\subsection{Linear Bounds}

We now return to arbitrary positive service rates $\rateSA$ and $\rateSB$ and construct two linear bounds for the equilibrium threshold. 
The upper bound is obtained by removing future arrivals of $\typeA$-customers, whereas the lower bound is obtained by considering the strategy under which $\typeB$-customers never renege.




\begin{theorem}
\label{thm:linear-bounds}
For every $\genb\ge0$, the equilibrium threshold function $\eqthreshB(\argdot)$ satisfies
\begin{equation}
\label{eq:linear-bounds}
\eqthreshB^-(\genb)\le \eqthreshB(\genb)\le \eqthreshB^+(\genb),
\end{equation}
where
\begin{align}
\label{eq:upper-linear}
\eqthreshB^+(\genb)
&\coloneqq
\bracks*{
\floor*{
\frac{\rateSA \rewardB}{\CostB}
-\frac{\rateSA}{\rateSB}(\genb+1)
}
}_{-1}^{\naor}
\intertext{and}
\label{eq:lower-linear}
\eqthreshB^-(\genb)
&\coloneqq
\bracks*{
\floor*{
\frac{\rateSA \rewardB}{\CostB \sumrho_{\naor}}
-\frac{\rateSA}{\rateSB}(\genb+1)
}
}_{-1}^{\naor}.
\end{align}
\end{theorem}

\begin{proof}
We first establish the upper bound.
Consider an auxiliary system in which no new $\typeA$-customers arrive, all other parameters remain unchanged, and $\typeB$-customers play an equilibrium threshold strategy.
Removing $\typeA$-arrivals can only improve Bernard's optimal continuation payoff, because  such arrivals may delay his service but cannot improve his position.
In the auxiliary system, Bernard's position over time can only improve, so he has no incentive to renege after joining. 
In state $(\nAi,\nBi)$, in equilibrium, Bernard's expected time until service is
\begin{equation}
\label{eq:expected-time-Bernard}   
\frac{\nAi}{\rateSA}+\frac{\nBi+1}{\rateSB}.
\end{equation}
Consequently, in state $(\nAi,\nBi)$ Bernard does not renege if and only if
\begin{equation}
\label{eq:expected-reward-Bernard}
\rewardB-\CostB\parens*{\frac{\nAi}{\rateSA}+\frac{\nBi+1}{\rateSB}}\ge 0.
\end{equation}
Solving this inequality for $\nAi$ and recalling that $\nAi\le \naor$ gives $\eqthreshB(\genb)\le \eqthreshB^+(\genb)$.

We now establish the lower bound.
Suppose that Bernard never reneges.
Because other  $\typeB$-customers may renege, Bernard's expected service time under this strategy is weakly smaller than the expected service time $\Etime^{\NR}$ in a system where no $\typeB$-customer reneges.

Recall that, by \cref{eq:s-d,eq:skgeneq}, 
\begin{equation}
\label{eq:s-NA-j} 
\sumrhok =  
\begin{dcases}
\frac{1-\utilizationA^{\threshgen+1}}{1-\utilizationA} & \text{if } \utilizationA\neq 1, \\ \threshgen+1 & \text{if } \utilizationA=1.
\end{dcases}
\end{equation}
Hence, 
\begin{equation}
\label{eq:sum-s-NA-j}  
\sum_{j=1}^{\threshgen}\sumrho_{\naor-j}
=
\begin{dcases}
\frac{\threshgen}{1-\utilizationA}
-\frac{\utilizationA^{\naor-\threshgen+1}(1-\utilizationA^{\threshgen})}
{(1-\utilizationA)^2} & \text{if } \utilizationA\neq1,\\   
\frac{\threshgen(2\naor+1-\threshgen)}{2} & \text{if } \utilizationA=1.
\end{dcases}
\end{equation}
Combining \cref{le:E-time,eq:s-NA-j,eq:sum-s-NA-j}, we get
\begin{equation}
\label{eq:Eno-reneg}   
\Etime^{\NR}(\nAi,\nBi)=
\frac1{\rateSA}
\sum_{j=1}^{\nAi}\sumrho_{\naor-j}
+\frac{\nBi+1}{\rateSB}\sumrho_{\naor}.
\end{equation}

Because $\sumrhok\le \sumrho_{\naor}$ for every $\threshgen\le \naor$, we obtain
\begin{equation}
\label{eq:E-NR<=}   
\EtimeNR(\nAi,\nBi)\le
\sumrho_{\naor}\parens*{
\frac{\nAi}{\rateSA}+\frac{\nBi+1}{\rateSB}}.
\end{equation}
Therefore, a sufficient condition for Bernard to remain in the queue in state $(\nAi,\nBi)$ is
\begin{equation}
\label{eq:Bernard-remain}   
\rewardB-\CostB \sumrho_{\naor}\parens*{
\frac{\nAi}{\rateSA}+\frac{\nBi+1}{\rateSB}} \ge 0.
\end{equation}
Solving for $\nAi$ yields $\eqthreshB^-(\genb)\le \eqthreshB(\genb)$, completing the proof.
\end{proof}

The two bounds in \cref{thm:linear-bounds} are truncated linear functions with the same slope $-\rateSA/\rateSB$. 
Their intercepts differ by
\begin{equation}
\label{eq:bound-distance}
\Delta\coloneqq
\frac{\rateSA \rewardB}{\CostB}
\left(1-\frac1{\sumrho_{\naor}}\right).
\end{equation}
Furthermore,
\begin{equation}
\label{eq:1-1-s-N-A}
1-\frac1{\sumrho_{\naor}}
=\utilizationA\frac{\sumrho_{\naor-1}}{\sumrho_{\naor}}
\le\utilizationA,
\end{equation}
where $\sumrho_{-1} \coloneqq 0$.
Hence,
\begin{equation}
\label{eq:delta-bound}
    \Delta\le\frac{\rateA \rewardB}{\CostB}.
\end{equation}
The distance between the two bounds provides an explicit sufficient condition for a good linear approximation.

\begin{corollary}
\label{cor:linear-approximation}
Define
\begin{equation}
\label{eq:linear-midpoint}
\widetilde \eqthreshB(\genb)\coloneqq
\bracks*{
\floor*{
\frac{\rateSA \rewardB}{2\CostB}
\parens*{1+\frac1{\sumrho_{\naor}}}
- \frac{\rateSA}{\rateSB}(\genb+1)}}_{-1}^{\naor}.
\end{equation}
Then
\begin{equation}
\label{eq:linear-error}
\sup_{\genb\ge0}
\abs*{\eqthreshB(\genb)-\widetilde \eqthreshB(\genb)}
\le
\ceil*{\frac{\Delta}{2}}
\le
\ceil*{\frac{\rateA \rewardB}{2\CostB}}.
\end{equation}
\end{corollary}

\begin{proof}
By Theorem~\ref{thm:linear-bounds}, $\eqthreshB(\genb)$ lies between $\eqthreshB^-(\genb)$ and $\eqthreshB^+(\genb)$.

The truncation function
$x\mapsto[x]_{-1}^{\naor}$ is nondecreasing and $1$-Lipschitz. 
Taking integer parts changes the difference between two arguments
by at most the ceiling of their distance.

Since the affine function defining $\widetilde \eqthreshB$ lies  between those
defining $\eqthreshB^-$ and $\eqthreshB^+$, we obtain
\[
|\widetilde \eqthreshB(\genb)-\eqthreshB^{\pm}(\genb)|
\le\left\lceil\frac{\Delta}{2}\right\rceil.
\]
Together with $\eqthreshB^-(\genb)\le \eqthreshB(\genb)\le \eqthreshB^+(\genb)$, this proves the first inequality. The second follows
from Eq.~\eqref{eq:delta-bound}.
\end{proof}

Corollary~\ref{cor:linear-approximation}
provides a sufficient condition directly
in terms of the model parameters.
In particular, if  $\rateA \rewardB/\CostB \le2$,
then the equilibrium threshold admits a
truncated linear approximation with a
uniform error of at most one.

More generally, if
$\rateA \rewardB/\CostB\le2m$ for some
positive integer $m$, the approximation
error is at most $m$.

The condition is sufficient but not
necessary.
Thus, although linear approximations can
perform arbitrarily poorly in general,
Theorem~\ref{thm:linear-bounds} provides
an explicit parameter-dependent guarantee
for their accuracy.

%
%

\section*{Acknowledgments}
The authors thank Rafi Hassin for his insightful comments and his bibliographic suggestions.

\bibliographystyle{apalike}
\bibliography{bib-multiclass-queues}

\appendix

%
%

\gdef\thesection{\Alph{section}} 
\makeatletter
\renewcommand\@seccntformat[1]{\appendixname\ \csname the#1\endcsname.\hspace{0.5em}}
\makeatother

\section{Gambler's Ruin}
\label{se:GamblerRuin}
This section is based on \citet[chapter~7]{Eth:Springer2010}.

Consider a bet that pays one unit of money. 
Let $p$ denote the probability of a win for the gambler, and $q=1-p$ the probability of a loss. 
Let $D$ denote the random variable representing the gambler's profit after a single bet. We have
\begin{equation*}
D=\begin{cases} 1 &\text{ with prob. } p \\ -1 &\text{ with prob. } q.
\end{cases}   
\end{equation*}
Let $(D_k)_{\threshgen\geq 1}$ be a sequence of \iid copies of $D$, that is, a sequence of independent bets. 
Then
\begin{equation*}
S_n=D_{1}+\dots+D_n    
\end{equation*}
represents the gambler's cumulative profit, with $S_{0}=0$. 
For $W,L \in \naturals$,  assume that the gambler stops
betting as soon as they win $W$ units or lose $L$ units. 
Hence, the gambler stops betting at the stopping time
\begin{equation*}
T(-L,W) \coloneqq \inf\braces*{n\geq0: \ S_n=-L \text{ or } S_n=W}.    
\end{equation*}
Then the gambler’s ruin probability and the expected duration of play are given as follows.

\begin{theorem} [\protect{\citealt[theorem~ 7.1.1]{Eth:Springer2010}}]
\label{theo:GamblerRuin1}
\begin{equation}
\label{eq:ProbGamblerRuin}
\Prob\parens*{S_{T(-L,W)}=-L}
=\begin{cases}
\displaystyle{1-\frac{\parens*{\frac{q}{p}}^L-1}{\parens*{\frac{q}{p}}^{L+W}-1}} & \text{ if } p\neq q, \\ 
\\ 
\displaystyle{\frac{W}{L+W}} & \text{ if } p= q. 
\end{cases}    \end{equation}
\end{theorem}

\begin{theorem}[\protect{\citealt[theorem~7.1.4]{Eth:Springer2010}}]
\label{theo:GamblerRuin2}
\begin{equation}
\label{eq:DurationGamblerRuin}
\Expect[T(-L,W)]
=\begin{cases}
\displaystyle{\frac{L+W}{q-p}\parens*{\frac{L}{L+W}-\frac{\parens*{\frac{q}{p}}^L-1}{\parens*{\frac{q}{p}}^{L+W}-1}}} & \text{ if } p\neq q,\\ 
\\ 
L\cdot W  & \text{ if } p=q.  \end{cases}
\end{equation}
\end{theorem}

Finally, the following corollary describes the limiting cases when either $L$ or $W$ tends to infinity.
\begin{corollary}[\protect{\citealt[corollary 7.1.6]{Eth:Springer2010}}]
    \label{cor:GamblersRuin}
    If p>q, then
    \begin{equation}
        \label{eq:DurationGamblerLinf}
        \Expect[T(-\infty,W)]=\frac{W}{p-q}.
    \end{equation}
     If p<q, then
    \begin{equation}
        \label{eq:DurationGamblerWinf}
        \Expect[T(-L,\infty)]=\frac{L}{q-p}.
    \end{equation}
\end{corollary}

\section{Additional Proofs} \label{se:appendix-proofs}

In this section, we recall some tools that are used for the proofs in \cref{se:equilibrium}.

Let Bernard's state be $(\nAi,\nBi)$. 
$\typeA$-customers follow the Naor threshold strategy $\naor$, and assume that $\typeB$-customers cannot renege. 
Using the convention that the state $(0,-1)$ represents the situation when Bernard has just been served, the transition probabilities for the queueing system are the following:
\begin{equation}
\begin{split}
&\text{for }
\nAi=0,\\
&(\nAi,\nBi) \to 
\begin{dcases}
(\nAi+1,\nBi)  &\text{with probability\ } \frac{\rateA}{\rateA+\rateSB},\\
(\nAi,\nBi-1)  &\text{with probability\ } \frac{\rateSB}{\rateA+\rateSB};
\end{dcases}  \\ 
&\text{for }0< \nAi<\naor,\\
&(\nAi,\nBi) \to 
\begin{dcases}
(\nAi+1,\nBi)  &\text{with probability\ } \frac{\rateA}{\rateA+\rateSA},\\
(\nAi-1,\nBi)  &\text{with probability\ } \frac{\rateSA}{\rateA+\rateSA};
\end{dcases}
\\
&\text{for }\nAi=\naor,\\
&(\nAi,\nBi) \to 
\begin{dcases}
(\nAi,\nBi)  &\text{with probability\ } \frac{\rateA}{\rateA+\rateSA},\\(\nAi-1,\nBi)  &\text{with probability\ } \frac{\rateSA}{\rateA+\rateSA}.  
\end{dcases} 
\end{split}
\end{equation}

Note that, given a state $(\nAi,\nBi)$, the state transitions only horizontally whenever $\nAi>0$, and may transition downwards vertically only if $\nAi=0$. 
The state never moves upwards.

Before computing Bernard's expected service time, we recall the following result on the absorption times of a one-dimensional Markov chain with an absorbing boundary and a partially reflecting boundary.

\begin{lemma}
[\protect{\citealt[eq.~(8) and (9)]{HarSwe:SR1970}}]
\label{le:UtimeFullyStrat}
Let $(\ncust_{\threshgen})_{\threshgen\geq0}$ be a random walk on the state space $\braces*{0,1,\dots,L}$, with initial state $\ncust_{0}=\gena$, and  transition probabilities given by \cref{fig:ruin-4}. 
Let $\Utime_{\gena}$ be the expected time until the random walk reaches $0$.
Then, for $0\leq\gena\leq L$, 
\begin{equation}
\label{eq:Utime}
\Utime_{\gena} =\begin{cases} \displaystyle{\frac{\rateA+\rateSA}{\rateSA-\rateA}\bracks*{\gena-\frac{\utilizationA^{L+1}}{1-\utilizationA}\parens*{\utilizationA^{-\gena}-1}}} & \text{ if } \utilizationA\neq 1,\\
\\
\gena\parens*{2L+1-\gena} & \text{ if } \utilizationA=1.
\end{cases}
\end{equation}  
\end{lemma}

\begin{figure}[h!]
\begin{center}
\begin{tikzpicture}[->, >=stealth, auto, semithick, node distance=2cm]
\tikzstyle{every state}=[fill=white,draw=black,text=black]

\node[] (A) {0};
\node[] (B) [right of=A] {1};
\node[] (C) [right of=B] {};
\node[] (D) [right of=C] {$\gena$};
\node[] (E) [right of=D] {};
\node[] (F) [right of=E] {$L-1$};
\node[] (G) [right of=F] {$L$};

\path (B) edge [bend left, below] node {$\frac{\rateSA}{\rateA+\rateSA}$} (A)
      (B) edge [dashed, above] node {} (C)
      (C) edge [bend left, above] node {$\frac{\rateA}{\rateA+\rateSA}$} (D)
      (D) edge [bend left, below] node {$\frac{\rateSA}{\rateA+\rateSA}$} (C)
      (D) edge [bend left, above] node {$\frac{\rateA}{\rateA+\rateSA}$} (E)
      (E) edge [bend left, below] node {$\frac{\rateSA}{\rateA+\rateSA}$} (D)
      (E)  edge [dashed, above] node {} (F)
      (F) edge [bend left, above] node {$\frac{\rateA}{\rateA+\rateSA}$} (G)
      (G) edge [bend left, below] node {$\frac{\rateSA}{\rateA+\rateSA}$} (F)
      (G)  edge [loop right] node {$\frac{\rateA}{\rateA+\rateSA}$} (G);

\end{tikzpicture}
\end{center}
\caption{Bounded Random-walk}
  \label{fig:ruin-4}
\end{figure}

\begin{lemma}
\label{le:E-time}
Consider a priority queueing system where $\typeB$-customers are not allowed to renege.  
Then, in state $(\nAi,\nBi)$, the expected time spent by Bernard in the system is
\begin{equation}
\label{eq:Etime}       \EtimeNR(\nAi,\nBi)=
\begin{dcases} 
\frac{1}{\rateSA(1-\utilizationA)} \bracks*{(\nBi+1)\frac{\rateSA}{\rateSB}\parens*{1 - \utilizationA^{\naor+1}}+  \nAi-\frac{\utilizationA^{\naor-\nAi+1}}{1-\utilizationA} \parens*{1-\utilizationA^{\nAi}}} 
& \text{ if } \utilizationA\neq1, \\
\\
\frac{\nAi}{2\rateSA}(2\naor+1-\nAi)+(\nBi+1)\frac{1}{\rateSB}(1+\naor) & \text{ if } \utilizationA=1.
\end{dcases}
\end{equation}

\end{lemma}

\begin{proof}
Notice that
\begin{equation}
\label{eq:EtimePreempFS}
\begin{split}
\EtimeNR(\nAi,\nBi)&=\EtimeNR(\nAi,\nBi-1)+\EtimeNR(0,0)\\
&=\EtimeNR(\nAi,0)+(\nBi+1)\EtimeNR(0,0)\\
&=\frac{\Utime_{\nAi}}{\rateA+\rateSA}+(\nBi+1)\EtimeNR(0,0).   
\end{split}
\end{equation}
Then, by  \cref{le:UtimeFullyStrat}, we have

\begin{equation}
\label{eq:UtimeFullStrat}
\Utime_{\nAi} =\begin{cases} \displaystyle{\frac{\rateA+\rateSA}{\rateSA-\rateA}\bracks*{\nAi-\frac{\utilizationA^{\naor+1}}{1-\utilizationA}\parens*{\utilizationA^{-\nAi}-1}}} & \text{ if } \utilizationA\neq 1,\\
\\
\nAi\parens*{2\naor+1-\nAi} & \text{ if } \utilizationA=1.
\end{cases}
\end{equation}

To compute $\EtimeNR(0,0)$ we need to consider that $\nBi$ decreases by $1$ with probability $\rateSB/(\rateSB+\rateA)$, and $\nAi$ increases by $1$ with probability $\rateA/(\rateSB+\rateA)$. 
Then,
\begin{equation}
\label{eq:E(0,1)}
\begin{split}    \EtimeNR(0,0)&=\frac{\rateSB}{\rateA+\rateSB}\frac{1}{\rateA+\rateSB}+\frac{\rateA}{\rateA+\rateSB}\parens*{\frac{1}{\rateA+\rateSB}+\EtimeNR(1,0)}\\
&=\frac{1}{\rateA+\rateSB}+\frac{\rateA}{\rateA+\rateSB} \parens*{\frac{\Utime_{1}}{\rateA+\rateSA}+\EtimeNR(0,0)},
\end{split}
\end{equation}
from which we get 
\begin{align}
\EtimeNR(0,0)\parens*{1-\frac{\rateA}{\rateA+\rateSB}}
&=\frac{1}{\rateA+\rateSB}+\frac{\rateA}{\rateA+\rateSB}\frac{\Utime_{1}}{\rateA+\rateSA}, \\
\EtimeNR(0,0)
&=\frac{1}{\rateSB}+\frac{\rateA}{\rateSB}\frac{\Utime_{1}}{\rateA+\rateSA}.
\end{align}
Hence, by \eqref{eq:UtimeFullStrat}, if $\utilizationA\neq 1$, we have

\begin{equation}\label{Eq:E-(0,0)-1}
\begin{split}
\EtimeNR(0,0)& = \frac{1}{\rateSB}+\frac{\rateA}{\rateSB(\rateSA-\rateA)} \parens*{1 - \frac{\utilizationA^{\naor+1}(\utilizationA^{-1}-1)}{1-\utilizationA}}\\
&=
\frac{1}{\rateSB}+\frac{\rateA}{\rateSB(\rateSA-\rateA)}\frac{1 -\utilizationA -\utilizationA^{\naor}+\utilizationA^{\naor+1}}{1-\utilizationA}\\
&=
\frac{1}{\rateSB}+\frac{\rateA}{\rateSB(\rateSA-\rateA)}\frac{(1-\utilizationA)\parens*{1-\utilizationA^{\naor}}}{1-\utilizationA}\\
&=
\frac{\rateSA-\rateA+\rateA\parens*{1-\utilizationA^{\naor}}}{\rateSB(\rateSA-\rateA)}\\
&=
\frac{\rateSA\parens*{1-\utilizationA^{\naor+1}}}{\rateSB(\rateSA-\rateA)}.
\end{split}
\end{equation}
Plugging \eqref{Eq:E-(0,0)-1} and \eqref{eq:UtimeFullStrat} into \eqref{eq:EtimePreempFS}, for $\utilizationA\neq 1$, we obtain
\begin{equation*}
\EtimeNR(\nAi,\nBi)= \frac{1}{(\rateSA-\rateA)} \bracks*{(\nBi+1)\frac{\rateSA}{\rateSB}\parens*{1 - \utilizationA^{\naor+1}}+  \nAi-\frac{\utilizationA^{\naor-\nAi+1}}{1-\utilizationA} \parens*{1-\utilizationA^{\nAi}}}.    
\end{equation*}
If $\utilizationA=1$,  we have
\begin{equation}
\label{eq:E-(0,0)-2}
\begin{split}
\EtimeNR(0,0)&=\frac{1}{\rateSB}+\frac{2\rateA\naor}{2\rateA\rateSB} \\
&=
\frac{\naor+1}{\rateSB},
\end{split}
\end{equation}
and
\begin{equation*}
\EtimeNR(\nAi,\nBi)=\frac{\nAi(2\naor+1-\nA)}{2\rateSA}+(\nBi+1)\frac{\naor+1}{\rateSB}.
\qedhere
\end{equation*}
\end{proof}


 

\end{document}